\documentclass[journal]{IEEEtran}
\usepackage{multirow, booktabs}
\usepackage{amsmath}
\usepackage{fancybox}
\usepackage{multirow}
\usepackage{amsfonts}
\usepackage{mathtools}
\usepackage{amsthm}
\usepackage{mathdesign}
\usepackage{MnSymbol}
\usepackage{graphicx}
\usepackage[table]{xcolor}
\usepackage{pbox}

\usepackage{url}

\theoremstyle{definition}
\newtheorem{definition}{Definition}[section]
\newtheorem{theorem}{Theorem}[section]
\newtheorem{lemma}{Lemma}[section]

\newcommand{\ignore}[1]{}
\usepackage{enumitem}
\usepackage{array}
\newcolumntype{P}[1]{>{\vspace*{0.07cm}\centering\arraybackslash}p{#1}}
\newcolumntype{L}[1]{>{\vspace*{0.07cm}\raggedright\arraybackslash}m{#1}}
\newcolumntype{M}[1]{>{\vspace*{0.07cm}\centering\arraybackslash}m{#1}}

\DeclareMathAlphabet{\mathrmbf}{\encodingdefault}{\rmdefault}{bx}{n}

\newcommand{\mbb}{\mathbb}

\renewcommand{\theenumi}{\Alph{enumi}}

\begin{document}
%
% paper title
% Titles are generally capitalized except for words such as a, an, and, as,
% at, but, by, for, in, nor, of, on, or, the, to and up, which are usually
% not capitalized unless they are the first or last word of the title.
% Linebreaks \\ can be used within to get better formatting as desired.
% Do not put math or special symbols in the title.
\title{ReHand: Secure Region-based Fast Handover with User Anonymity for Small Cell Networks in 5G}
%
%
% author names and IEEE memberships
% note positions of commas and nonbreaking spaces ( ~ ) LaTeX will not break
% a structure at a ~ so this keeps an author's name from being broken across
% two lines.
% use \thanks{} to gain access to the first footnote area
% a separate \thanks must be used for each paragraph as LaTeX2e's \thanks
% was not built to handle multiple paragraphs
%

\author{Chun-I~Fan~\IEEEmembership{Member,~IEEE}, Jheng-Jia~Huang, Min-Zhe~Zhong, Ruei-Hau~Hsu$^{\ast}$~\IEEEmembership{Member,~IEEE}, Wen-Tsuen~Chen,~\IEEEmembership{Fellow,~IEEE}, and Jemin~Lee~\IEEEmembership{Member,~IEEE}% <-this % stops a space
\thanks{C.-I.~Fan is with the Department
of Computer Science and Engineering, National Sun Yat-sen University, Kaohsiung 80424,
Taiwan. E-mail: cifan@mail.cse.nsysu.edu.tw}% <-this % stops a space
\thanks{J.-J.~Huang is with the Department of Computer Science and Engineering, National Sun Yat-sen University, Kaohsiung, Taiwan. E-mail: jhengjia.huang@gmail.com}
\thanks{M.-Z.~Zhong is with the Department
of Computer Science and Engineering, National Sun Yat-sen University, Kaohsiung 80424,
Taiwan. E-mail: stjoe410262@gmail.com}
\thanks{R.-H.~Hsu is with Institute for Infocomm Research~(I2R), Agency for Science, Technology and Research~(A*STAR), Singapore, 138634. E-mail: richard\_hsu@i2r.a-star.edu.sg}
\thanks{W.-T. Chen is with the Institute of Information Science, Academia Sinica, Taiwan, Taipei 11529, Taiwan and Department of Computer Science, National Tsing Hua University, Hsinchu 30013, Taiwan. E-mail: chenwt@iis.sinica.edu.tw.}
\thanks{J. Lee is with the Department of Information and Communication
Engineering, Daegu Gyeongbuk Institute of Science and Technology, Daegu, 43016
South Korea.  E-mail: jmnlee@dgist.ac.kr.}
\thanks{The contact author is R.-H.~Hsu.}
% <-this % stops a space
%\thanks{Manuscript received April 19, 2005; revised September 17, 2014.}
}

\maketitle

% As a general rule, do not put math, special symbols or citations
% in the abstract or keywords.
\begin{abstract}
Due to the expectedly higher density of mobile devices and exhaust of radio resources, the fifth generation~(5G) mobile networks introduce small cell concept in the radio access technologies, so-called {\bf Small Cell Networks}~(SCNs), to improve radio spectrum utilization. However, this increases the chance of handover due to smaller coverage of a micro base station, i.e., home eNodeB~(HeNB) in 5G. Subsequently, the latency will increase as the costs of authenticated key exchange protocol, which ensures entity authentication and communication confidentiality for secure handover, also increase totally. Thus, this work presents a secure region-based handover scheme~(ReHand) with user anonymity and fast revocation for SCNs in 5G. ReHand greatly reduces the communication costs when UEs roam between small cells within the region of a macro base station, i.e., eNB in 5G, and the computation costs due to the employment of symmetry-based cryptographic operations. Compared to the three elaborated related works, ReHand dramatically reduces the costs from 82.92\% to 99.99\%. Nevertheless, this work demonstrates the security of ReHand by theoretically formal proofs.

\end{abstract}

% Note that keywords are not normally used for peerreview papers.

% For peer review papers, you can put extra information on the cover
% page as needed:
% \ifCLASSOPTIONpeerreview
% \begin{center} \bfseries EDICS Category: 3-BBND \end{center}
% \fi
%
% For peerreview papers, this IEEEtran command inserts a page break and
% creates the second title. It will be ignored for other modes.
\IEEEpeerreviewmaketitle

\section{Introduction}
% The very first letter is a 2 line initial drop letter followed
% by the rest of the first word in caps.
% 
% form to use if the first word consists of a single letter:
% \IEEEPARstart{A}{demo} file is ....
% 
% form to use if you need the single drop letter followed by
% normal text (unknown if ever used by IEEE):
% \IEEEPARstart{A}{}demo file is ....
% 
% Some journals put the first two words in caps:
% \IEEEPARstart{T}{his demo} file is ....
% 
% Here we have the typical use of a "T" for an initial drop letter
% and "HIS" in caps to complete the first word.
 \IEEEPARstart{F}{ifth} {generation~(5G) mobile communication networks play as a key role in not only communication technologies, but also the internet-of-things~(IoT) technologies. It has the holistic enhancement in the radio access technologies~(RATs) and communication infrastructures for the new data/information exchange requirements in IoT. For example, millimeter wave, device-to-device communications, and small cell networks, etc., are to provide better quality of services~(QoS), utilization of radio resource and bandwidth, higher transmission rate, and lower latency for various of emerging applications~\cite{5G-PPP_proposal, Woon_2014, Demestichas_2013, Hong_2014, Olsson_2013}.}

To achieve this objective, we need a more efficient network architecture with robust security to meet the demands of application performance in next decade. To meet those demands, research projects~\cite{website:METIS, 5G-PPP_proposal, techreport:Ericsson_5G_research} have focused on 5G mobile communication networks that provide a flexible, reliable, and high-performance network architecture for wireless communication beyond 2020. Owing to the upcoming IoT~\cite{Atzori_2010} and the increase in the use of mobile devices~\cite{techreport:Ericsson_mobility_report}, a considerably higher capacity of wireless networks is required. In the future, connections on the wireless system will increase rapidly and will be more complex. 5G also emphasizes much lower latency and higher data rate for users. By decreasing the latency, we can improve the stability of data transmission and provide real-time services. Increasing the data rate for each user enables more advanced applications, such as high-definition (HD) mobile television (mobile-TV) and mobile clouds.

{Several notable technologies have been addressed in 5G, including support of IPv6, Flat-IP based network, pervasive networks, power efficiency technology, massive machine-type communications, and small cell networks, etc~\cite{website:5G-PPP,Janevski_2009,Tudzarov_2011,techreport:Ericsson_mobility_report,5G-PPP_proposal,3GPP-HNB-Arch,3GPP-HNB-Mobility,3GPP-HNB-Sec}.}  {In recent years, several researches and technical reports surveyed a wide range of information of 5G~\cite{Felita_2013, Gohil_2013, techreport:Ericsson_5G_research, Patel_2012, Singh_2012}. These provided several diverse collections of 5G features and their own comments.} As a result, 5G is an important future trend in the next decade, and it will be established in 2020.

Our research focuses on one of these main technical trends, known as small cell networks~(SCNs), a new concept of infrastructures under the macro cellular coverage~\cite{website:small_cell}. Types of small cells include femto-, pico-, and micro-cells, which provide different levels of coverage and abilities. The smallest cells, i.e., femto cells, have a coverage area of an office, whereas the pico and micro cells have coverage of a building and a community, respectively. Macro cells do not communicate with user terminals directly but focus on the management and connectivity of small cells in an urban scope.

\subsection{Small Cell Networks}
The basic architecture of {small cell networks in 5G} consists of {user equipment~(UE), eNodeB/Home eNodeB~(eNB/HeNB), security gateway~(SeGW), home subscribe server~(HSS), and authentication center~(AuC)~\cite{3GPP-HNB-Arch,3GPP-HNB-Mobility,3GPP-HNB-Sec}. A UE obtains the communication services via HeNB/eNB using RATs. Mobility management entity~(MME) manages the mobility of UEs by processing the handover requests and updating the tracking areas of UEs.} At the very beginning, each HeNB has to perform mutual authenticated key exchange with SeGW for establishing a secure channel among them. A handover process occurs when a UE changes its visiting eNB/HeNB when the signal from the connected eNB/HeNB becomes weaker. {When the handover process occurs, the UE and the new eNB/HeNB should authenticate the mutual legality and exchange a session key for the following secure communications. For better utilization of radio resources, small cell technology in 5G deploy more base stations with smaller coverage, i.e., HeNBs. Compared to the service capacity of macro cell, each small cell can serve the same number of UEs in a smaller coverage, so that the density of UEs is enhanced.}

Access control in 5G is essential to provide the functions correctly in the system, i.e., both UE and infrastructure should be able to identify if the counterpart of each other is legal or not. Besides, the confidentiality of the subsequent communications is also required. Generally, access control and secure communication are guaranteed by authenticated key exchange~(AKE) protocols. The design of AKE in SCN is more challenging as its performance requirement is more critical since the latency might increase significantly due to higher chance of handover caused by the smaller coverage of radio access networks.  Hence, a new design of AKE to reduce the costs of AKE in handover is necessary to fulfill the performance requirements of supporting real-time applications in 5G. Additionally, the user anonymity to conceal the footprint of communications should be considered to guarantee privacy as more and more personal and sensitive information is involved in applications.

\subsection{Related Work}
A notable amount of roaming-based AKE protocols have been proposed and user anonymity has been carefully deliberated in~\cite{Wireless_Auth_ZM04,Wireless_Auth_JLSS06,Wireless_Auth_TO08,Wireless_Auth_YHWD10,Wireless_Auth_HBCCY11,Wireless_Auth_HCCB12,Wireless_Auth_RH13,RoamAuth_LLLLZS14,Wireless_Auth_GH15,Wireless_Auth_HCG15,RoamAuth_LCCHAZ15,RoamAuth_HCG15}. In mobile networks, an UE should complete authentication for identity identification prior to requesting for services when roaming to the coverage of a new visiting foreign network~(FN). The user anonymous authentication prevents eavesdroppers or/and FN from exposing the real identities of UEs in every authentication session such that the footprints of communications of UEs are concealed.

User anonymity can be separated into two levels, partial user anonymity and full user anonymity. Partial user anonymous authentication conceals identities from eavesdroppers, excluding FNs~\cite{Wireless_Auth_ZM04,Wireless_Auth_JLSS06,Wireless_Auth_TO08} and full user anonymous authentication additionally considers FNs as eavesdroppers ~\cite{Wireless_Auth_YHWD10,Wireless_Auth_HBCCY11,Wireless_Auth_HCCB12,Wireless_Auth_RH13,RoamAuth_LLLLZS14,Wireless_Auth_HCG15,RoamAuth_LCCHAZ15,RoamAuth_HCG15}. With full user anonymity, traceability and revocability are essential to support the permitted network operators to trace and revoke user identities for management purposes. Diverse traceability and revocability techniques~\cite{Wireless_Auth_YHWD10,Wireless_Auth_HBCCY11,Wireless_Auth_HCG15} have been developed to manage the anonymity protection in roaming-based mobile networks. However, in order to provide strong user anonymity, the costs of revocation and tracing are commonly considerably high in certain roaming-based AKE schemes. In~\cite{RoamAuth_HCG15}, the system revokes the users by updating user private keys periodically. In~\cite{RoamAuth_LCCHAZ15}, a time-bound user anonymity AKE is proposed to reduce the costs of revocation checking by eliminating the revoked users, whose credentials expire naturally. Overall, the aforementioned elegant works resolves privacy protection requirements for roaming-based AKE. However, the revocation and tracing costs for management purposes might be enlarged for SCNs in 5G. Thus, an efficient design of roaming-based AKE for secure handover with user anonymity is urgently required to fit the features of SCNs.

%\IEEEPARstart{M}{obile} communication networks have played a very important role in our lives since they emerged in the early 1990s. Most significant is the fact that each new generation of mobile networks provides much better quality of service (QoS) than its predecessors. Wireless data transfer technologies, which include voice, picture, video transmission, and network surfing, support the development of information services in recent decades.

%\begin{figure}[!t]
%  	\centering
%  	\includegraphics[scale=0.267]{figs/f1.eps}
%  	\caption{The 5G mobile communication networks}
%  	\label{fig:f1}
%\end{figure}	

\begin{figure*}[!htb]
	\begin{minipage}{.5\textwidth}
		\includegraphics[width=\textwidth]{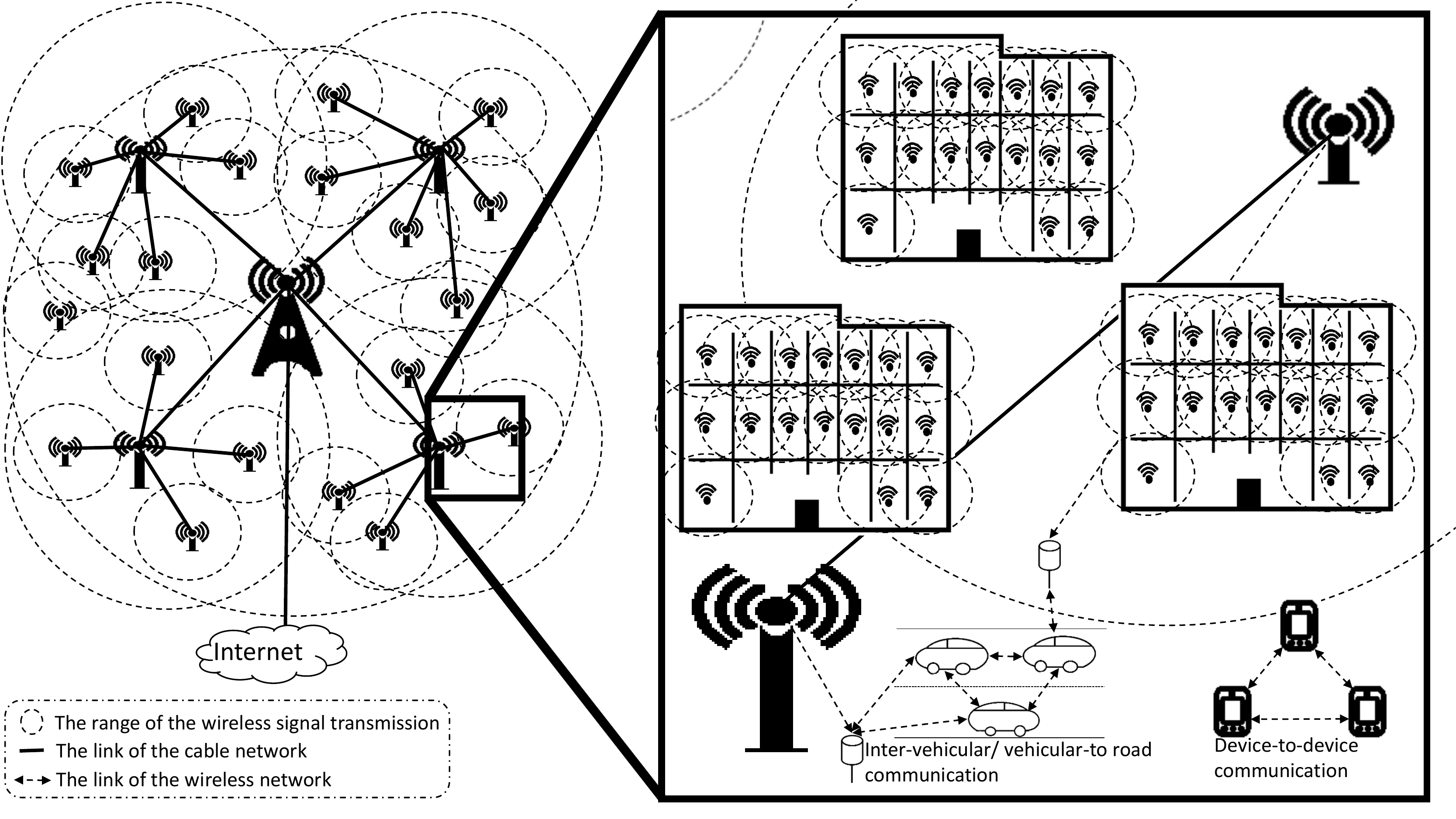}
		\caption{The 5G mobile communication networks}
		\label{fig:f1}
	\end{minipage}
	\begin{minipage}{.5\textwidth}
		%\centering
		\includegraphics[width=\textwidth]{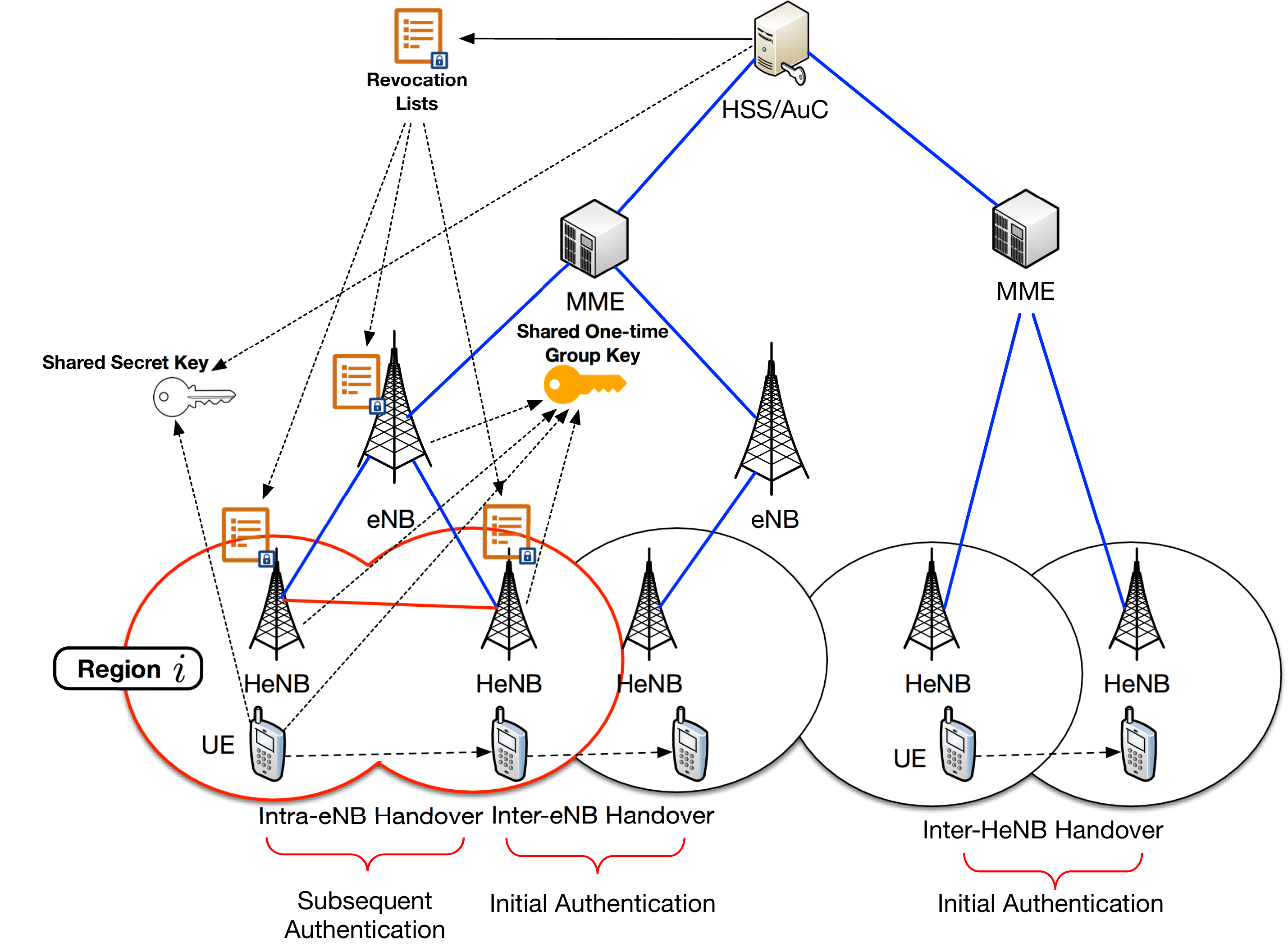}
		\caption{The System Model of Small Cell Networks in 5G.}
		\label{fig:sys_model_5G}
	\end{minipage}
\end{figure*}

%Put the figures of 5G mobile communication networks and system model of small cell network in 5G together

\subsection{Difference between 4G and 5G}
There are significant differences between the current 4G and 5G communication networks such as connection principles and infrastructures. In 4G, the {UE} should connect to the macro cell base station, i.e., the eNB, during the handover process to an adjacent {macro cell}~\cite{Akhtar_2011}. Depending on the service area of the {eNB or HeNB}, {UEs} in 4G {might} connect to the {eNB} directly if their signals can be detected by the {eNB}. {UEs} can enhance their connections to the macro cell through relay nodes (RNs), which are a variation of base stations deployed at the coverage edge of macro cells in 4G. In this manner, the 4G coverage can be enhanced.

In 5G, {UEs} cannot connect to the macro cells directly because the macro cells do not broadcast beacon frames anymore. The user terminals in 5G can only connect to the macro cell via the help of the {HeNBs in small cells}~\cite{Akhtar_2011}. The infrastructure of 5G includes a massive increase in the number of {HeNBs}.

%Despite the benefits, deploying a large number of {HeNBs} causes some problems. For instance, the frequency of the {required} handover process in 5G increases rapidly along with the number of increased {HeNBs}. However, if we follow the connection principles in 4G, the total latency will increase since {UEs} communicate with the macro cell whenever they run the handover process. This does not meet the requirement of ultra-low latency in 5G. 
Therefore, the operational performance {in small cell technology} of the original 4G design is inefficient when it {applies} to the new architecture of 5G. 
Besides, a vulnerable handover procedure will {suffer from the higher risks of impersonation and eavesdropping attacks due to more infrastructural components in 5G}.
In order to {enhance} the performance and guarantee the security of {small cell networks} in 5G, {it is essential to} design a {fast} handover authentication {mechanism}, which is secure, {efficient}, and tailored to {the design} principles of 5G.

\subsection*{Contributions}
In this paper, we propose a {region-based} secure handover authentication scheme {for small cell network in 5G with the following contributions: 1) the design of region-based fast authentication reduces the communication and computation costs without involving the components of core networks when roaming to a new visiting micro cell within the same region of a macro cell; 2) user anonymous authentication of roaming to a new macro cell guarantees identity anonymity against disclosure of communication footprints; 3) user membership revocation by accumulated one-way hash eliminates the costs of managing all revoked users in 5G system; 4) this work proposes a metrics to evaluate the performance of region-based fast handover authentication compared to the other related works; 5) the paper proves the security of the proposed scheme fulfilling the security definitions by theoretical proofs.}

\section{{System and Security Models}}
{This section introduces the proposed system model, the security requirements, the security model, the corresponding security definitions of 5G small cell networks.}
\subsection{{System Model}}
{This section introduces the system model of small cell networks, including the proposed security architecture and mobility, by referring to 3GPP 5G standards~\cite{3GPP-5G-Sec_Arch}. The entire 5G network consists of RAN and evolved packet core~(EPC), also known as core network. In 5G RAN, there are different types of RANs regarding their transmission coverage, transmission power, service capacity, and application scenario~(e.g., indoor or outdoor), such as macrocell, microcell, picocell, and femtocell. Macrocell is the RAN of the largest coverage supported by eNBs, and microcell, picocell, and femtocell are the RANs of smaller coverage supported by HeNBs.  As depicted in Fig.~\ref{fig:sys_model_5G}, a UE may attach to an eNB or HeNB, for services provided by the core network, i.e., EPC, in mobile networks. A eNB attaches to EPC directly and HeNB may attach to an eNB or EPC directly. Before serving the UE, both UE and EPC have to complete AKE to verify the legitimacy of both parties and establish a shared session key to secure the following communications. The entire AKE procedure involves UE, eNB/HeNB, MME, and HSS/AuC. The UE and the HSS/AuC shares a common long-term secret key. Hence, in AKE, the MME will request the HSS/AuC to generate the required authentication token and verify the authentication messages from the UE, respectively.  The UE can generate authentication token and verify the authentication messages from the MME by the shared secret with the HSS/AuC. Once the AKE is completed, both UE and MME will share the same secret key material. The MME will also send the derived session keys to the eNB/HeNB for the subsequent secure communications with the UE.}

{{\bf Handover.}} {For traditional mobility model to small cell networks, when the UE roams to a new HeNB, it has to perform a complete AKE with the components of EPC. This naturally increases the latency of communications, especially more handovers in small cell networks. Hence, the mobility model of the proposed scheme defines a {\it Region}, which is formed by an eNB and its belonging HeNBs. The eNB and HeNBs within the same region will share the same group secret key for performing region-based fast handover AKE. Thus, the entire authentication with MME and HSS/AuC is required, when the UE roams to a new visiting region. The UE only need to perform fast handover AKE, when the UE roams to a HeNB within the visited region without involving the MME and HSS/AuC.}

%\begin{figure}[!t]
%	\centering
%	\includegraphics[scale=0.32]{figs/5G_System_Model.eps}
%	\caption{The System Model of Small Cell Networks in 5G.}
%	\label{fig:sys_model_5G}
%\end{figure}	

\subsection{{Security Requirements}} 
{
\begin{itemize}
	\item{\bf Authenticated Key Exchange:}  Before mobile services, the 5G security system should ensure identity identification through mutual authentication between UE and the system components in RAN and EPC. Additionally, both UE and 5G security should be able to exchange a session key securely to protect the subsequent communications. 
	\item{\bf Identity Anonymity:} The identity anonymity guarantees that any two communication sessions from the same UE is unlinkable to any outsider eavesdropper. Hence, the identity of each UE for every communication session should be randomized to avoid the traceability of footprint of communications. 
	\item{\bf Fast Authentication:} In order to enhance the performance of authentication, UE performs authentication protocol with only nearest component, e.g., HeNB, eNB, etc. It should be able to reduce the latency caused by the communications with and computation on MME and HSS/AuC in the core network. 
	\item{\bf Active/Passive Revocation:} Revocation is an essential function to ensure that the subscription of each UE can be revoked in case of expiration or suspension of services. In the proposed 5G security system, the system will issue a temporary group key for fast authentication when the UE roams to the coverage of new serving eNB. The temporary group key can be revoked passively when it is expired. The system can also revoke the given temporary group keys of the user by issuing revocation lists for the specified eNBs and HeNBs.
	\item{\bf Traceability:} In order to locate UEs for certain services, such as incoming calling services and short message service, the system should be able to identify the location of each UE even if the anonymous identity is used to conceal the footprints of communications.
\end{itemize}
}

\subsection{Security Definitions}

\begin{definition}
	Matching Conversations~\cite{Bellare_1993_Auth} \\
	The proposed protocol $\Pi$ in the presence of an adversary $E$ and consider two oracles, $\Pi^s_{A,B}$ and $\Pi^t_{B,A}$, that model two entities $A$ and $B$ being the partners of each other in the communication sessions $s$ and $t$. We say that $\Pi_{A,B}^{s}$ and $\Pi_{B,A}^{t}$ have match conversation if and only if $s=t$ and $A$ and $B$ are partners. 
\end{definition}

\begin{definition}
	$No - Matching_E(k)$~\cite{Bellare_1993_Auth} \\
	Let $No - Matching_E(k)$ be the event that there exists $A, B, s, t$ such that $\Pi^s_{A,B}$ accepted, but there is no oracle $\Pi^t_{B,A}$ which engaged in a matching conversation under the presence of a polynomial time adversary $E$. Note that $k$ is a security parameter, $A, B \in I$, and $s, t \in N$.
\end{definition}

\begin{definition}
	Secure Mutual Authentication ~\cite{Bellare_1993_Auth} \\
	We say that $\Pi$ is a secure mutual authentication protocol if for any polynomial time adversary $E$, 
	\begin{itemize}
		\item[$(1)$] if oracles $\Pi^s_{A,B}$ and $\Pi^t_{B,A}$ have matching conversations, then both oracles accept.
		\item[$(2)$] the probability of $No - Matching_E(k)$ is negligible.
	\end{itemize}
\end{definition}

\begin{definition}
	$Distinguish_{sk_E}(k)$~\cite{Bellare_1993_Auth}\\
	Let $Distinguish_{sk_E}(k)$ be the event that an adversary $E$ can correctly guess that she/he is given the real session key or a random number after the protocol is performed and terminates successfully, where $k$ is a security parameter.
\end{definition}

\begin{definition}
	A secure mutual authentication and key exchange protocol ~\cite{Bellare_1993_Auth} \\
	A protocol $\Pi$ is a secure mutual authentication and key exchange protocol if the following properties are satisfied:
	\begin{itemize}
		\item[1.] $\Pi$ is a secure mutual authentication protocol.
		\item[2.] $\Pi^s_{A,B}$ and $\Pi^t_{B,A}$ hold the same session key after running $\Pi$ successfully.
		\item[3.] (Indistinguishability) : (The probability of $Distinguish_{sk_E}(k) - \dfrac{1}{2}$) is negligible.
	\end{itemize}
\end{definition}

\begin{definition}
	\label{def:IND-CCA}
	The game for INDistinguishability under the Chosen-Ciphertext Attack (IND-CCA)~\cite{Goldwasser_1984}\\
	A challenger $\psi$ and a polynomial time adversary $\Gamma$ play the following game with a symmetric cryptosystem $\Pi$.
	\begin{itemize}
		\item{\bf Step 1.}  $\psi$ runs a setup algorithm. $\psi$ gives $\Gamma$ the resulting public parameters $params$. An encryption oracle $E_{sk}$ and the decryption oracle $D_{sk}$ are given a key $sk$. The above oracles hold the secret key secretly.
		\item{\bf Step 2.} $\Gamma$ issues a sequence of encryption and decryption queries. Upon receiving an encryption query, denoted by $m^{\ast}$, $\psi$ returns $\pi^{\ast} = E_{sk}(m^{\ast})$ to $\Gamma$. Upon receiving a decryption query, denoted by $\pi^{\ast}$, $\psi$ returns $\rho^{\ast} = D_{sk}(\pi^{\ast})$ to $\Gamma$.
		
		{\bf Challenge: } $\Gamma$ outputs a plaintext pair $(m_0, m_1)$. Upon receiving $(m_0, m_1)$, $\psi$ randomly chooses $\theta \in \{0, 1\}$ and computes the ciphertext $\pi = E_{sk}(m_{\theta})$. Then, $\psi$ returns $\pi$ to $\Gamma$.
		\item{\bf Step 3.} $\Gamma$ issues a sequence of encryption and decryption queries as those in {\bf Step 2} where a restriction here is that $\pi^{\ast} \neq \pi$.
		
		{\bf Guess:} Finally, $\Gamma$ outputs $\theta' \in \{0, 1\}$. If $\theta' = \theta$, $\Gamma$ will win the game.
	\end{itemize}
	The polynomial time adversary $\Gamma$ participated in the game is referred to as an IND-CCA adversary with the guessing advantage $Adv^{IND-CCA}_{\Pi}(\Gamma) = \vert Pr[\theta = \theta'] - \dfrac{1}{2} \vert$.
\end{definition}

\begin{definition}
	IND-CCA Security \\
	We can say that a symmetric cryptosystem is $(t, \varepsilon)$-IND-CCA secure if no polynomial time adversary $\Gamma$ within running time $t$, has guessing advantage $Adv^{IND-CCA}_{\Pi}(\Gamma) \geq \varepsilon$ after performing the game of Definition II.6.
\end{definition}

\begin{definition}
	The game for indistinguishability under a pseudorandom permutation and a random permutation (PRP)~\cite{Song_2000}\\ 
	A challenger $\psi$ and a polynomial time adversary $\Gamma$ play the following game with a pseudorandom permutation $\Omega$.
	\begin{itemize}
		\item{\bf Step 1.} $\psi$ runs a setup algorithm. $\psi$ gives $\Gamma$ the resulting public parameters $params$. There are two oracles, $\Omega$ and $\Omega^{-1}$, which are the pseudorandom permutation and its inverse, respectively. $\Omega$ can be regarded as an encryption function and $\Omega^{-1}$ can be regarded as the decryption function. $\Omega$ and $\Omega^{-1}$ know a secret key $k$. $\omega$ and $\omega^{-1}$ are the random permutation and its inverse, respectively. The random permutation $\omega$ is regarded as an encryption function and $\omega^{-1}$ is regarded as the decryption function.
		\item{\bf Step 2.} $\Gamma$ issues a sequence of $\Omega_k$ and $\Omega_k^{-1}$ queries. Upon receiving a $\Omega_k$ query, denoted by $\rho^{\ast}$, $\psi$ returns $\pi^{\ast} = \Omega_k(\rho^{\ast})$ to $\Gamma$. Upon receiving a $\Omega_k^{-1}$ query, denoted by $\pi^{\ast}$, $\psi$ returns $\rho^{\ast} = \Omega_k^{-1}(\pi^{\ast})$ to $\Gamma$.
		
		{\bf Challenge:}\\
		{\bf Case I:} 
		$\Gamma$ sends a plaintext $\rho$ to $\psi$ with a restriction that $\rho$ is different from each $\rho^{\ast}$ in {\bf Step 2}. $\psi$ randomly chooses $\theta \in \{0, 1\}$ and computes $\pi = \Omega_k(\rho)$ when $\theta = 0$ or $\pi = \omega(\rho)$ when $\theta = 1$. Then, $\psi$ returns $\pi$ to $\Gamma$. 
		
		{\bf Case II:} 
		$\Gamma$ sends a ciphertext $\pi$ to $\psi$ with a restriction that $\pi$ is different from each $\pi^{\ast}$ in {\bf Step 2}. $\psi$ randomly chooses $\theta \in \{0, 1\}$ and computes $\rho = \Omega_k^{-1}(\pi)$ when $\theta = 0$ or $\rho = \omega^{-1}(\pi)$ when $\theta = 1$. Then, $\psi$ returns $\rho$ to $\Gamma$.
		
		\item{\bf Step 3.} $\Gamma$ issues a sequence of $\Omega_k$ and $\Omega_k^{-1}$ queries as those in {\bf Step 2} where restrictions here are that $\rho^{\ast} \neq \rho$ and $\pi^{\ast} \neq \pi$.
		
		{\bf Guess:} Finally, $\Gamma$ outputs $\theta' \in \{0, 1\}$. If $\theta' = \theta$, $\Gamma$ will win the game.
		
		The polynomial time adversary $\Gamma$ participated in the game is referred to as a PRP adversary with the guessing advantage $Adv^{PRP}_{\Omega}(\Gamma) = |Pr[\Gamma^{\Omega_k, \Omega^{-1}_k} = 1] - Pr[\Gamma^{\omega, \omega^{-1}} = 1]| = |Pr[\theta' = \theta]-\dfrac{1}{2}|$.
	\end{itemize}
\end{definition}

\begin{definition}
	Pseudorandom Permutation Security (PRP Security)\\
	If no polynomial time adversary $\Gamma$ within running time $t$, has the advantage $Adv^{PRP}_{\Omega}(\Gamma) \geq \varepsilon$ after performing the game of II.8, then the function $\Omega : K_{\Omega} \times Z \rightarrow Z$ can be considered as a $(t, \varepsilon)$-secure pseudorandom Permutation~\cite{Song_2000}. Note that $K_{\Omega}$ is the key space of key $k$ and $Z = \{0, 1\}^n$ where $n$ is a security parameter. 
	
	The guessing advantage of $\Gamma$ is $Adv^{PRP}_{\Omega}(\Gamma) = |Pr[\Gamma^{\Omega_k, \Omega^{-1}_k} = 1] - Pr[\Gamma^{\omega, \omega^{-1}} = 1]| = |Pr[\theta' = \theta]-\dfrac{1}{2}|$. Note that $\omega$ is a random permutation selected uniformly from the set of all bijections on $Z$, and $k$ is chosen randomly from the set of key space $K_{\Omega}$.
\end{definition}

\begin{definition}
	Pseudorandom Function Security (PRF Security)\\
	If no polynomial time adversary $\Gamma$ within running time $t$, has the advantage $Adv^{PRF}_{\Lambda}(\Gamma) \geq \varepsilon$ after performing the game of II.8, then the function $\Lambda : K_{\Lambda} \times Z \rightarrow Z$ can be considered as a $(t, \varepsilon)$-secure pseudorandom Function~\cite{PRF_1996}. Note that $K_{\Lambda}$ is the key space of key $k$ and $Z = \{0, 1\}^n$ where $n$ is a security parameter. 
	
	The guessing advantage of $\Gamma$ is $Adv^{PRF}_{\Lambda}(\Gamma) = |Pr[\Gamma^{\Lambda_k} = 1] - Pr[\Gamma^{\lambda, \lambda^{-1}} = 1]| = |Pr[\theta' = \theta]-\dfrac{1}{2}|$. Note that $\lambda$ is a random function selected uniformly from the set of all bijections on $Z$, and $k$ is chosen randomly from the set of key space $K_{\Lambda}$.
\end{definition}

\section{Preliminaries}

\subsection{One-Way Accumulator}
\label{sec:accumulator}

One-way accumulator firstly introduced by J. Benaloh and M. de Mare in 1993~\cite{Benaloh_1993} is a one-way hash function with a quasi-commutative property for the purpose of testing membership without the help of a trusted authority.\\
%\textit{Definition:}\\
$\star${\bf Definition 1: one-way hash function~\cite{Benaloh_1993}}\\
A family of $One-way ~ hash ~ functions$ is an infinite set  of functions $h_l:X_l\times Y_l \rightarrow Z_l$ having the following properties:
\begin{description}[labelsep=0.5em]
\item 1. There exists a polynomial $P$ such that for each integer $l, h_l(x, y) $ is computable in time $P( l, |x|, |y|) $ for all $x_l \in X_l$ and all $y_l \in Y_l$.
\item 2. There is no polynomial $P$ such that there exists a probabilistic polynomial time algorithm which, for all sufficiently large $l$, will when given $l$, a pair $(x, y) \in X_l \times Y_l$, and a $y' \in Y_l$, find an $x' \in X_l$ such that $h_l(x, y) = h_l(x', y')$ with probability greater than $1/P(l)$ when $(x, y)$ is chosen uniformly among all elements of $X_l \times Y_l$ and $y'$ is chosen uniformly form $Y_l$.
\end{description}
$\star${\bf Definition 2: Quasi-commutativity~\cite{Benaloh_1993}}\\
A function $f: X \times Y \rightarrow X$ is said to be $quasi-commutative$ if for all $x \in X$ and for all $y_1, y_2 \in Y$, $f(f(x, y_1), y_2) = f(f(x, y_2), y_1)$.

\noindent $\star${\bf Definition 3: Nyberg's One-way accumulator~\cite{Nyberg_1996}}\\
A family of $one-way$ $accumulators$ is a family of one-way sh functions with quasi-commutativity. The one-way accumulator by K. Nyberg~\cite{Nyberg_1996} is constructed based on the generic symmetry-based hash function~(e.g., SHA) and simple bit-wise operations.  Compared to Benaloh's scheme~\cite{Benaloh_1993}, Nyberg's scheme is more efficient without employing asymmetric cryptographic operations. Assume that the upper bound to the number of accumulated items is $N = 2^d$ where $d$ is a positive integer and let one-way hash function $h:\{0,1\}^{*}\rightarrow{}\{0,1\}^{l=r\times{}d}$, where $r$ is a positive integer. Let $x_1, x_2,...,x_m$ be the accumulated items with different string sizes and $y_i$ is the hashing value for each $x_i$, such that $\{y_i = h(x_i)\}_{i\in [1, m]}$, where $m \leq N$. $y_{ij}$ can be represented as $y_i = (y_{i,1},...,y_{i,r})$, for $y_{ij}\in\{0,1\}^{d}$ and $j=1,...,r$. Next, we replace $y_{i,j}$ by a single bit. If $y_{i,j}$ is a string comprised of $d$ 0s, it is replaced by 0. Otherwise, $y_{i,j}$ is replaced by 1. Since there are $r$ substrings of $y_{i,j}$, $y_i$ can be mapped to a string $b_i\{0,1\}^{r}$, such that $b_i = (b_{i,1}, b_{i,2}, ..., b_{i,r}) = \alpha(y_i) = \alpha(h(x_i))$. The $b_{i,j}$ denotes the $j^{th}$ bit of $b_i$ and the probability of $b_{i,j} = 0$ is $2^{-d}$. In this way, we can transfer an accumulated item $x_i$ to a bit string $b_i$ of length $r$ which can be considered as a value of $r$ independent binary random variable if $h$ is an ideal hash function. Let $H^{Nyb}()$ denote Nyberg's fast one-way hash function and $\odot$ be the bitwise operation AND. The accumulated function on an accumulated item $X$ with an accumulated key $K$ can be described as $H^{Nyb}(K, X) = K \odot \alpha(Y) = K \odot \alpha(h(X))$. And it also can be described as $Z =H^{Nyb}(K, X) = K\odot\alpha(y_i)= K \odot \alpha(h(x_i))$ for $i=1,...,m$ if $X$ is a set of accumulated items $X = (x_1, x_2, ..., x_m)$. As the bitwise operation AND obeys the commutativity rule, the quasi-commutativity of $H^{Nyb}()$ can be achieved and $H^{Nyb}(H^{Nyb}(K, x_1), x_2)= H^{Nyb}(H^{Nyb}(K, x_2), x_1)$. On the other hand, the operation AND as a logic multiplication operation also has the property of absorbency, which can be expressed as "A $\odot$ A = A ". Hence, $H^{Nyb}(H^{Nyb}(K, x_i),x_i)= K \odot \alpha(h(x_i)) = H^{Nyb}(K, x_i)$. To verify the membership of an item $x_i$ on the accumulated value $Z$ expressed as $(a_1, a_2, ..., a_r)$, compute $b_i = \alpha(h(x_i))$ corresponding to $b_i = (b_{i,1}, b_{i,2}, ..., b_{i,r})$ and check that whenever $b_{i,j} = 0$ then $a_j = 0$ for all $j = 1, ...,r$.  Using the property of absorbency, one can verify whether an item $x_i$ within the accumulated value $Z$ by $H^{Nyb}(Z, x_i) = Z \odot \alpha(h(x_i)) = Z$. The security proof of Nyberg's one-way accumulator~\cite{Nyberg_1996} is based on the availability of a long, truly random hash code which provides strong one-wayness property. In other words, it can be proven secure in the Random Oracle Model~\cite{Bellare_1993_RO,Canetti_2004}.

\section{Proposed Scheme}
{The proposed region-based fast authentication introduces the concept of regional warrants, where each region is formed by the coverage of a macro cell, which includes one eNB and several belonging HeNBs.  A UE will be issued a regional warrant when visiting a new region and completing the {\bf Initial Handover} protocol. When the UE roams to another HeNB within the same region, the region-based fast handover authentication will be performed with the assistant of MME and HSS/AuC. Hence, the fast handover authentication greatly reduces the communication latency.}
%We import the concept of warrants into our proposed scheme. User terminals should get a warrant before joining the mobile network with the help of the core networks and the operator. With the help of the warrants, users can complete the handover process with each legal AP very quickly and smoothly. On the other hand, the warrant is linkable but untraceable so that the infrastructures can provide the services without knowing the real identity of any user. 
Our {protocol} also provides an active revocation function such that the operator can revoke {UE's} warrant actively when {the membership of UE is changed. }. The notations used in our {protocols} are shown in Table I. 

\subsection{Overview and Key Management}
The proposed scheme contains five phases, {\bf Initialization, Registration, Initial Handover, Region-based Fast Handover}, and {\bf Active Revocation}. First, the {\bf Initialization} phase introduces the initialization of the parameters of macro cells~{(eNB)}, small cells {(HeNBs)}, {the MME, and the HSS/AuC}. The {\bf Registration} phase presents the procedures of registering a new UE joining with its identity and security information in the mobile network. How a user joins {a new visiting region} and gets a warrant, which is generated by {the eNB of the region}, are shown in the {{\bf Initial Handover}} phase. The {\bf {Region-based Fast} Handover} phase presents that how a {UE} and the { visiting HeNB} authenticate each other when the user moves from one {small} cell to another. Finally, we present an active revocation function in the {\bf Active Revocation} phase. 

%\subsection{{Proposed Region-based Fast Handover Authentication}}
Let the {region covered by eNB} be the macro cell, {the regions covered by HeNBs} be the small cells, and UEs be the {mobile} user terminals. Table~\ref{tab:key_management} the variables of key management in the proposed ReHand scheme. 

\begin{table}[tbh]
	\centering
	\caption{Notations}
	\begin{tabular}{||c|l||}
	\hline\hline
	Notation & Meaning\\
	\hline
	{$GK_i$} & {group key of the eNB $i$ and the belonging HeNBs}\\\hline
	$ID_i$ & identity of {UE} $i$\\\hline
	$pID_i,rID_i$ & anonymous identity of UE $i$ \\\hline
	$TID_{ij}$& anonymous identity of UE $i$ for region $j$ \\\hline
	 $bR^{I}_j$& blind factor of anonymizing identity\\\hline
	%$Sig_x$ & signer $x$'s signing function\\\hline
	\multirow{2}{1.8em}{$E_{x}(y)$} & using a symmetric encryption function to encrypt\\& message $y$ with key $x$ \\\hline
	\multirow{2}{1.8em}{$D_{x}(y)$} & using a symmetric decryption function to decrypt\\& message $y$ with key $x$ \\\hline
	$K_{i}$ & {shared} long-term secret key between UE $i$ and HSS/AuC\\\hline
	%$K_{A-O}$ & long-term secret key between the AAA server and\\& the operator\\\hline
	$d$ & one-time key between {UE} $i$ and {HSS/AuC}\\\hline
	\multirow{2}{1.8em}{$D_{ij}$} & {region} secret key shared between {UE} $i$ and the region \\ 
	         & {covered by the eNB $j$ and its belonging HeNBs.}\\\hline
	\multirow{2}{1.8em}{$R^{S_t}_{j}$} & the accumulated value of the revocation list for the \\ 
	      &  region $j$  in time slot $S_t$ \\\hline
	$H^{Nyb}$ & Nyberg's fast one-way hash function\\\hline
	$H,F$ & one-way hash functions\\\hline
	$T_{ex}$ & timestamp of warrant's expiration time\\\hline
	%$T_{\alpha}$ & timestamp of the sending time generated by $\alpha$\\\hline
	%$T_{cur}$ & timestamp of current time\\\hline
	%$\theta$ & acceptable latency \\\hline
	$\odot$ & bitwise operation AND\\\hline
	\hline
	\end{tabular}
	\label{tab:notation}
\end{table}

\begin{table}[tbh]
	\centering
	\caption{{Key Management in ReHand}}
	\begin{tabular}{||c|c|c||}
		\hline\hline
		UE & eNB\&HeNBs of region $j$ & HSS/AuC\\
		\hline
		$K_i$, $D_{ij}$ &$GK_{j}$& $\{D_{ij}\}_{i\in[1,N]}$,$\{GK_j\}_{j\in[1,M]}$ \\
		$TID_{ij}$ & $R^{S_t}_{j}$  & $\{K_{i}\}_{i\in[1,N]},\{rID_{i}\}_{i\in[1,N]}$\\
		$pID_i$&$\{bR^{I}_j\}_{I\in[1,k]}$ &$K_{H},\{pID_i\}_{i\in[1,N]}$
		\\\hline
		\hline
	\end{tabular}
	\label{tab:key_management}
\end{table}

%%==================================================
%%==================================================

%$\star${\bf Initialization:}\\
\subsection{{Initialization}}
This phase produces the required parameters as follows:
	\begin{description}
	\item[-]{\bf Step 1:}
	%In order to secure the data transfer between the AAA server and APs, a group key $GK$ is generated in the system. The group key $GK$ can be generated by using an existing key agreement scheme. Only the legal AAA server and APs can obtain $GK$, which will be kept in secret. 
	{In order to form a region of fast handover, the HSS/AuC issues a group key $GK_{j}$, blind factors $\{bR^{I}_j\}_{I\in[1,k]}$ for identity anonymization, to each eNB $j$ and its belonging HeNBs.}
	\item[-]{\bf Step 2:}
	{The HSS/AuC issues a long-term secret key $K_i$ for each UE $i$ as the shared long-term secret between them. It also selects an anonymous identity $rID_i$ and computes $pID_i=E_{K_H}(rID_i)$ for each UE $i$.}
	%The operator generates its own public key and private key, and then keeps the private key in secret and publishes the public key.
	%\item[-]{\bf Step 3:}
	%The AAA server and the operator share a long-term secret key $K_{A-O}$ for secure communication.
	\item[-]{\bf Step 3:}
	{The HSS/AuC prepares a revocation list $R^{S_t}_j$ for each region $j$ in time slot $S_t$, where the region secret keys of the revoked UEs are accumulated by Nyberg's accumulated hash function. $R^{S_t}_j$ is empty initially.}
	%The AAA server initializes the revocation list $R$.
	\end{description}
	
%%==================================================	
	
%$\star${\bf Register:}\\
\subsection{{Registration}}
{UE $i$ registers to the system securely} and share the private parameters with the HSS/AuC.
	\begin{description}
	\item[-]{\bf Step 1:}
	{UE} $i$ registers to the system with the identity information $ID_i$ and the registration required information.
	\item[-]{\bf Step 2:}
	The {HSS/AuC issues} $ID_i, rID_i, pID_i$, and $K_{i}$ to the registered UE $i$ and records $(ID_i,rID_i,pID_i)$ in its database.
	%\item[-]{\bf Step 3:}
	%MN $i$ keeps two pseudorandom identities $rID_i, rID^{\ast}_i$ and sets $rID^{\ast}_i = rID_i$.
	\end{description}

In practical situations, a {mobile user} should contact with the system operator in person before consuming the service. {The mobile user will be issued a SIM card, which is temper-resistant and stores the personal identity and secret, i.e.,$ID_i$, and $K_{i}$,} for her/his mobile device.

%%==================================================

%$\star${\bf Initial Handover:}\\
\subsection{{Initial Handover}}
When {a UE} $i$ {roams to the coverage of a new HeNB belonging to a newly visiting eNB $j$}, it should process the following steps to get a timeliness warrant from the system. The UE $i$ can, therefore, access the network {by running} the {region-based} handover process with each HeNB in the same coverage of the eNB with the warrant until it expires.

%Since {UEs} cannot communicate with the operator or the AAA server directly, the data from MN $i$ will be transferred through AP. 
In order to preserve {identity} privacy, the real identity $ID_i$ of {UE} $i$ should be hidden during transferring data. The {HSS/AuC} and eNBs/HeNBs take $pID_i$ as a label so that they can {extract the corresponding $K_i$ and $rID_i$ for the following authentication and key exchange}.

	\begin{description}
	 \item[-]{\bf Step 1:}
	 {UE} $i$ chooses a one-time key $d$ at random and sends $\{pID_i, C_1=E_{K_i}(pID_i,d)\}$ to the new visiting HeNB.
	 \item[-]{\bf Step 2:}
	 {When the HeNB received $\{pID_{i},C_1\}$, it forwards them to the MME through the eNB. Once the MME received $(pID_i,C_1)$, it sends them to HSS/AuC for authentication.} 
	 \item[-]{\bf Step 3:}
	 After receiving the messages from {the MME}, the {HSS/AuC first retrieves $K_i$ by the corresponding $rID_i=D_{K_H}(pID_i)$ and $d$ by decrypting $C_1$ with $K_i$. The HSS/AuC then checks if the decrypted $pID_i$ is equal to the $pID_i$ sent by the UE $i$. It then selects a new anonymous identity $rID_i$~(replace the original $rID_i$ with it) and $I\in\{1,k\}$, computes $TID_{ij}=\{\lambda=rID_i\oplus{}bR^{I}_j\}$, computes $pID^{*}_i=E_{K_H}(rID_i)$, $D_{ij}=H(GK_j,rID_i,T_{ex})$, $C_2=E_{K_i}(TID_{ij},D_{ij},T_{ex},d,pID^{*}_{i})$, $CK=H(d,pID^{*}_i)$, $K_{M,i}=F(CK)$, keeps the session key, $CK$, shared with the UE $i$, and sends $\{C_2,K_{M,i}\}$ back to the MME.}
	 \item[-]{\bf Step 4:}
	 After receiving the messages from the {HSS/AuC}, the {MME computes the session key, $K_{eN,i}=F(K_{M,i})$, shared with the eNB and the UE $i$, and sends $\{C_2,K_{eN,i}\}$ back to eNB. The eNB keeps $K_{eN,i}$ and sends the session key, $K_{He,i}=F(K_{eN,i})$, shared with the UE $i$.}  {Finally, both eNB and HeNB  share $K_{eN,i}$ and $K_{He}$ with the UE $i$, respectively, for the subsequent secure communications. Afterward, the HeNB sends $C_2$ to UE $i$.}

	 \item[-]{\bf Step 5:}
	 Upon the receipt of the messages from the {HeNB}, the {UE $i$} extracts $\{TID_{ij},D_{ij},T_{ex},$ $d,$ $pID^{*}_i\}$ by decrypting $C_2$ with $K_i$. It then replaces $pID_i=pID^{*}_i$, updates $TID_{ij}$, and computes the session keys shared with the HSS/AuC, MME, eNB, and HeNB by $CK=H(d,rID^{*}_i), K_{M,i}=F(CK), K_{eN,i}=F(K_{M,i}), K_{He,i}=F(K_{eN,i})$.
%Then the AAA server sends $C_4$ and a timestamp $T_{A}$ to MN $i$ where
%		\[C_4 = E_d\{C_r, D_i, Sig_O( rID_i ), T_{ex}, Q_4\}\]
%		\[Q_4 =  h(T_A)\].
  % \item[-]{\bf Step 6:}
%   After receiving the messages from the AAA server, MN $i$ performs the following operations.
 %  		\begin{itemize}
 %  		\item[-] Decrypt the ciphertext $C_4$ and $C_r$.
%		\[D_d(C_3) = \{C_r, D_i, Sig_O( rID_i ), T_{ex}, Q_4\}\]
%		\[D_{K_{M_i-O}}(C_r) = rID^{\ast}_i\]
%		\item[-] Check if the following formula is true.
%		\[h(T_A) = Q_4\]
%   		\item[-] Check the freshness of timestamp. 
%		\[T_{cur} - T_A < \theta ?\]
%   		\end{itemize}
  %\ignore{ 		
%\begin{figure}[!t]
%  	\centering
% 	\includegraphics[scale=0.35]{figs/Initial_Auth_SmallCell_5G_new.eps}
%  		\caption{Initial authentication for entering the coverage of a new attaching eNB}
%  	\label{fig:fig_MJ1}
%\end{figure}	 

%\begin{figure}[!t]
 % 	\centering
  %	\includegraphics[scale=0.37]{figs/Member_joining_part2.eps}
 % 		\caption{Initial Handover phase - Part 2}
  %	\label{fig:fig_MJ2}
%\end{figure}
%}	
\end{description}

%combining fig3 and fig4
\begin{figure*}[!htb]
	\begin{minipage}{.5\textwidth}
		\includegraphics[scale=0.32]{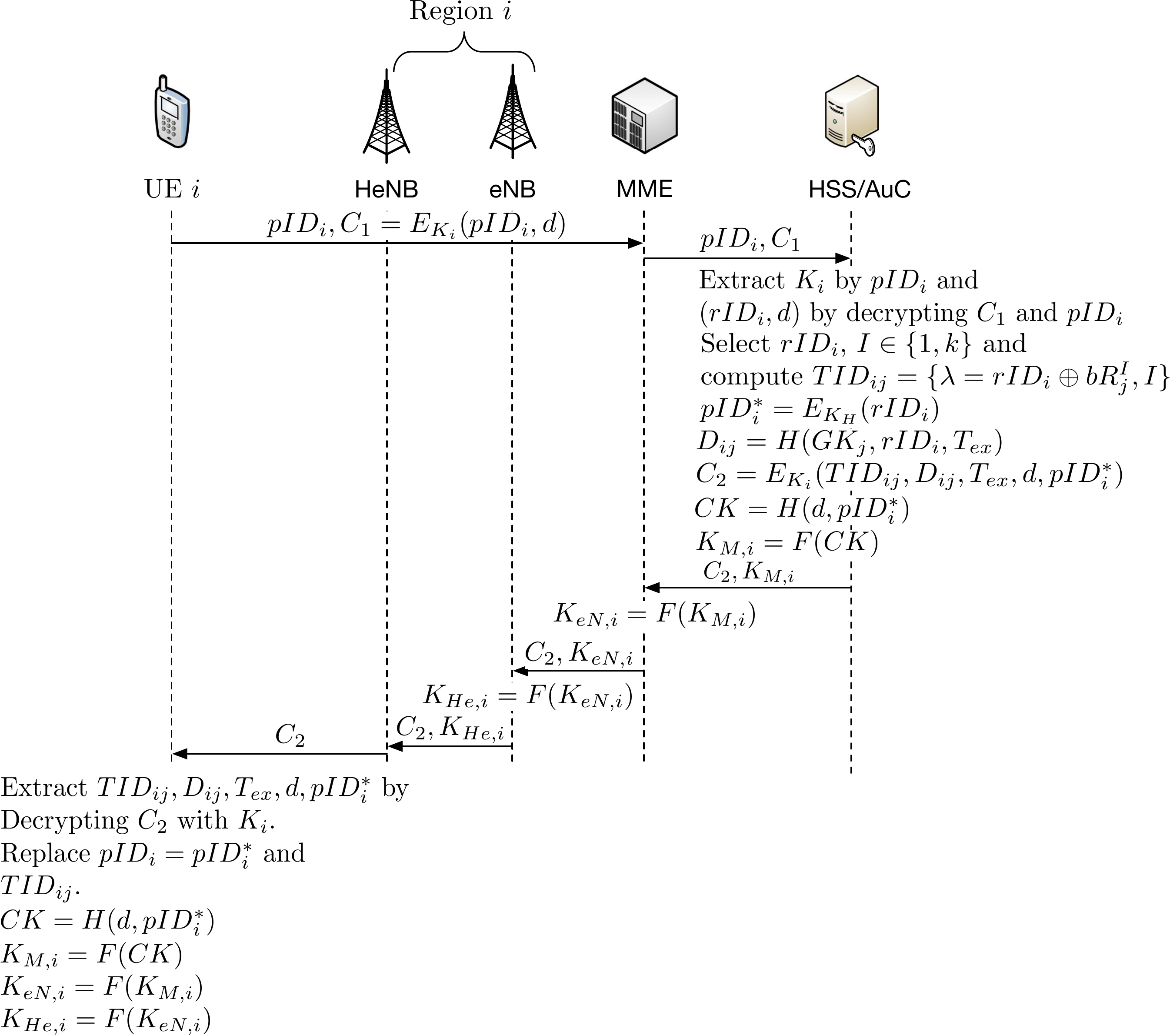}
		\caption{Initial authentication for entering the coverage of a new attaching eNB}
		\label{fig:fig_MJ1}
	\end{minipage}
	\begin{minipage}{.5\textwidth}
		%\centering
		\includegraphics[scale=0.37]{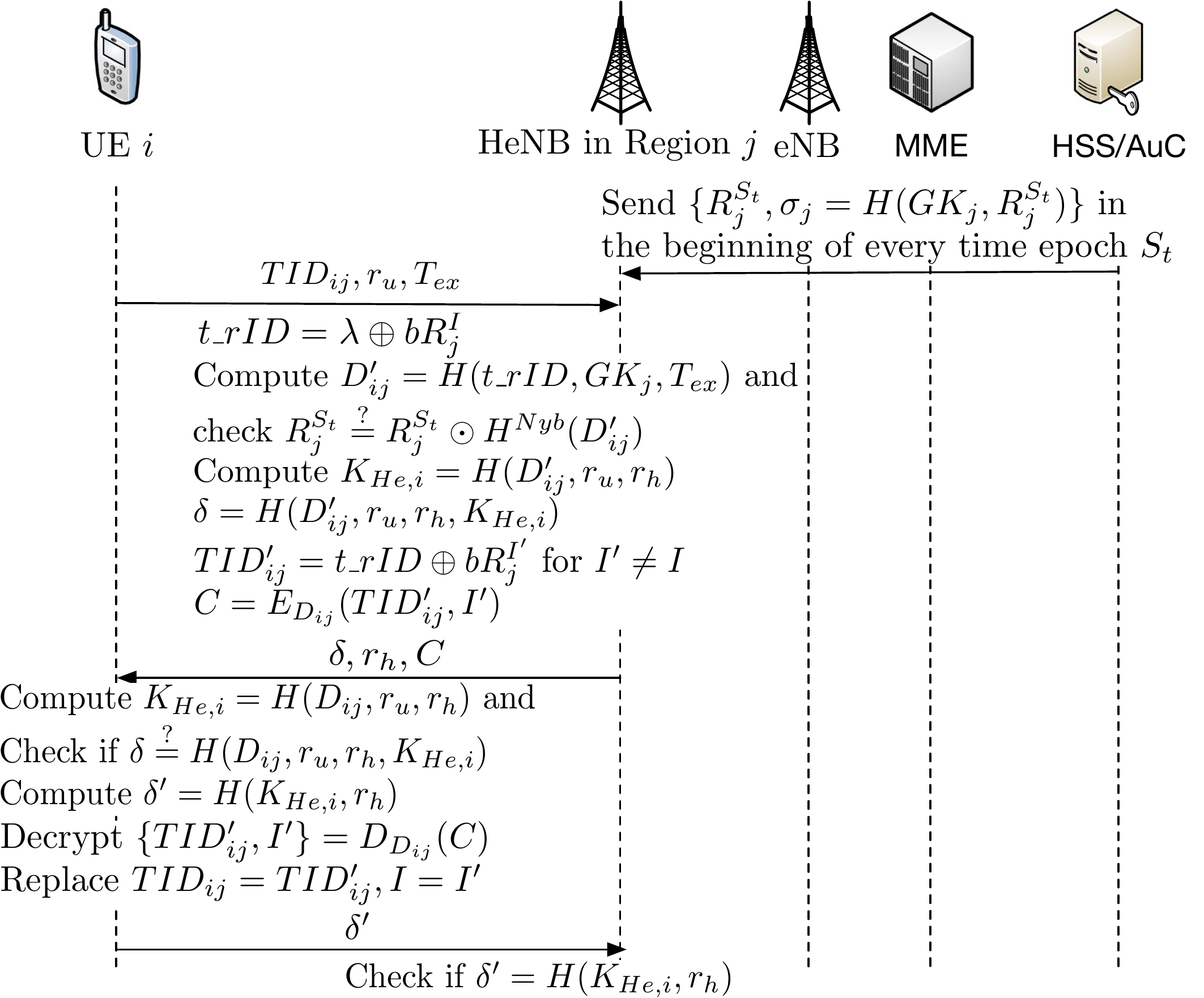}
		\caption{{Region-based} Handover phase}
		\label{fig:fig_HO}
	\end{minipage}
\end{figure*}
After the initial handover phase, the visiting region $j$ of the UE $i$ will be recorded by the MME. The UE $i$ will share $D_{ij}$ and $pID_i$, which will be updated for every initial handover session, with the HSS/AuC. If any UE $i$ is revoked, the HSS/AuC will update $R^{S_t}_{j}$ by $R^{S_t}_{j}=H^{Nyb}(R^{S_t}_{j}, D_{ij})$ and send to the HeNBs and the eNB of the region $j$.
%Since both of MN $i$ and the AAA server have authenticated with the operator, so that they have also authenticated each other indirectly base on the trust with the operator. In other words, MN $i$ and the AAA server authenticate each other with the help of the operator. As a result, we generate a legal warrant, a short-term key $D_i$, for MN $i$ in this phase.\\
	
%%==================================================	
	
%$\star${\bf {Region-based Fast} Handover:}\\
\subsection{{Region-based Fast Handover}}
When {UE} $i$ moves from one HeNB to a new HeNB within the same region $j$, it needs to run the {region-based} handover authentication. How {UE} $i$ and the {HeNB} authenticate each other securely and exchange a session key is described as follows.
%\ignore{

%\begin{figure}[!t]
%  			\centering
%  			\includegraphics[scale=0.38]{figs/5G_Fast_Handover_Auth_new.eps}
%  			\caption{{Region-based} Handover phase}
%  		\label{fig:fig_HO}
%\end{figure}
%}
			
\begin{description}
	\item[-]{\bf Step 0:} {The HSS/AuC updates $R^{S_{t}}_{j}$ for the revoked $D_{ij}$ in the previous time epoch $S_{t-1}$ and send it with the proof $\sigma_j=H(GK_j,R_{j}^{S_{t}})$ to the HeNB in the region $j$ at the beginning of $S_{t}$. Once receiving $\{R_{j}^{S_t},\sigma_j\}$, the HeNB check the correctness with $GK_{j}$ by $\sigma_j\stackrel{?}{=}H(GK_j,R_{j}^{S_t})$.
	}
	\item[-]{\bf Step 1:}
	{The UE $i$} sends {$\{TID_{ij},r_{u},T_{ex}\}$} to the {HeNB in the region $j$}. 
	\item[-]{\bf Step 2:}
	Upon the receipt of the messages from {UE} $i$, the {HeNB} extracts $bR^{I}_{j}$ by $I$ and computes  $t_rID=\lambda\oplus{}bR^{I}_j$ and computes $D'_{ij}=H(t_rID,GK_j,T_{ex})$. It then checks $R^{S_t}_{j}\stackrel{?}{=}H^{Nyb}(R^{S_t}_{j}, D'_{ij} )$ and $T_{cur} - T_{ex} \stackrel{?}{<} \theta$. If so, the HeNB selects a nonce $r_h$, and computes the shared session key $K_{H_{e,i}}=H(D'_{ij},r_u,r_h)$ and the response $\delta=H(D'_{ij},r_u,r_h,K_{H_{e,i}})$. It also computes $TID'_{ij}=t\_rID\oplus{}bR^{I'}_{j}$ by selecting a new $I'\neq{}I$ in $[1,k]$ and $C=E_{D_{ij}}(TID'_{ij},I')$. The HeNB sends $\delta,r_{h},C$ to the UE $i$.
	\item[-]{\bf Step 3:}
	Afterward, the UE $i$ computes $K_{H_{e,i}}=H(D_{ij},r_u,r_h)$ and checks if $\delta\stackrel{?}{=}H(D_{ij},r_u,r_h,K_{H_{e,i}})$. If so, the UE $i$ computes $\delta'=H(K_{H_{e,i}},r_h)$, $\{TID_{ij},I'\}=D_{D_{ij}}(C)$, and replaces $TID_{ij}=TID'_{ij}$ and $I=I'$. The UE $i$ then sends $\delta'$ to the HeNB.
	\item[-]{\bf Step 4:} { Upon the receipt of $\delta'$, the HeNB checks if $\delta'=H(K_{H_{e,i}},r_h)$. If it holds, the UE $i$ is legal and accepted.}
	\end{description}
If the protocol is completed, both UE $i$ and the HeNB are legal and accept each other. The handover {protocol} is also illustrated in Fig.~\ref{fig:fig_HO}.

%%==================================================	

%$\star${\bf Active Revocation:}\\
\subsection{{Active Revocation}}
{This phase shows the procedure of revoking a UE in the system}. {When the system revokes the membership of a UE}, {the HSS/AuC can revoke the issued warrant and long-term secret key as the following.}
	\begin{description}
	\item[-]{\bf Step 1:}
	The {system} operator provides {$ID_i$ of the revoked UE $i$} to the {HSS/AuC. The HSS/AuC will check all the temporarily anonymous identities and the corresponding warrants unexpired and issued for the UE $i$.} 
	\item[-]{\bf Step 2:}
	The {HSS/AuC} revokes the unexpired warrants by updating the $R^{S_t}_j = H^{Nyb}(R^{S_t}_j, D_{ij})$ of each region $j$ for time epoch $S_t$ with its message authentication code $\sigma_{j} = H(GK_j, R^{S_t}_{j})$. The HSS/AuC then sends each updated $\{R^{S_t}_j,\sigma_j\}$ to its belonging region.
	\item[-]{\bf Step 3:} Each HeNB will verify if the UE is revoked or not by the operations indicated in {\bf Step 2} of {\bf Region-based Fast Handover}. 
	\end{description}

%$\star${\bf The Setting of Time Slots:}
\subsection{Management of User Warrants and Revocation Lists}
	\begin{figure}[!t]
	\centering
	\includegraphics[scale=0.18]{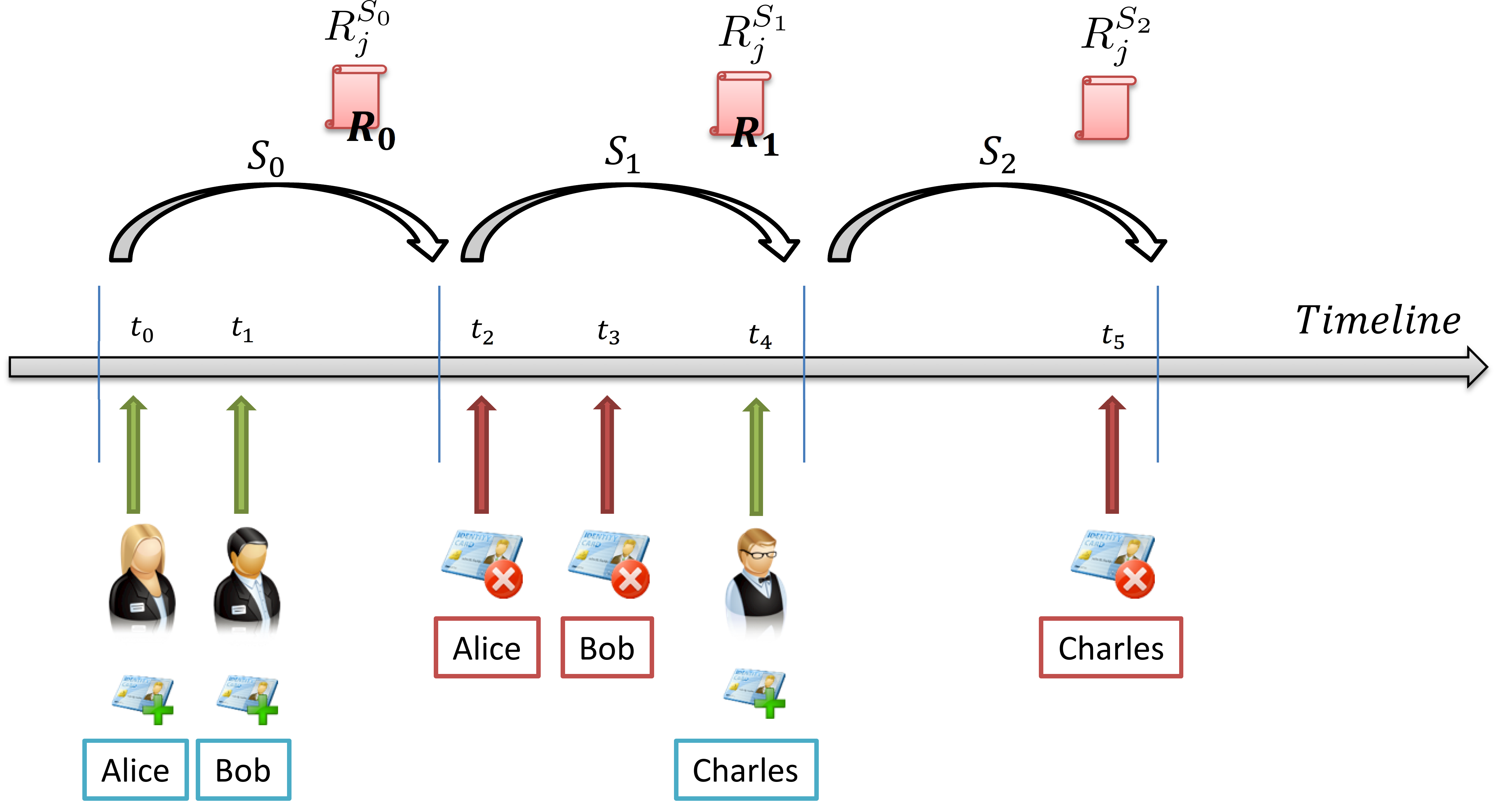}
	\caption{Analysis of the revocation list}
	\label{fig:fig_R}
\end{figure}	
Fig.~\ref{fig:fig_R}  illustrates the setting of time slots {for the update of revocation list of each macro cell region}. The {HSS/AuC} generates the accumulated value {$R^{S_t}_{j}$} as the revocation list {for macro cell region $j$} in each time slot {$S_{t}$ for all $t\in\mathbb{N}$}. Each $R_{j}^{S_t}$ contains the $D_{ij}$ of all revoked UEs in the region $j$. For example, Alice, Bob, and Charles get their warrants at $t_0, t_1, t_4$, respectively. Assume that all warrants are only valid for one time slot, the warrants of Alice, Bob, and Charles will expire at $t_2, t_3, t_5$, respectively. If the warrants of Alice and Bob are actively revoked before the expiration times $t_2, t_3$ and within $S_0$, {their} warrants {will} be accumulated into $R^{S_0}_{j}$.  If the active revocation time is before $t_2$ and $t_3$, and within $S_1$, the issuing time of $R^{S_1}_{j}$ is after $t_2$ and $t_3$.  Hence, the HSS/AuC needs not to accumulate the warrants of Alice and Bob into $R_{j}^{S_1}$ since the warrantes will be revoked passively by the expiration times.
  
Similarly, if the warrant of Charles is actively revoked before the expiration times $t_5$ and within $S_1$, the warrant of Charles should be accumulated into the revocation list {$R_{j}^{S_1}$}. Otherwise, Charles's warrant will be expired passively after $t_5$. The revocation list of each region in every time slot will be produced by the HSS/AuC and sent to the corresponding region for the active revocation of all the unexpired warrants.\\
%At the beginning of $S_2$, the AAA server drops $R_1$ and creates $R_3$. The AAA server can drop $R_1$ at the beginning of $S_2$ because all of the expired warrants accumulated in $R_1$ have been passively revoked and  all of the unexpired warrants accumulated in $R_1$ have been also accumulated in $R_2$. In this way, the AAA server only needs to store two revocation lists at any moment and thus, the cost of storage is constant.

\noindent{{\bf Security of Shared Group Key.} 
The eNB and HeNBs in the same region $j$ share the same group key $GK_j$ for the region-based fast handover authentication. Any eNB or HeNB may be accessed physically by adversaries, who intend to retrieve $GK_j$ and impersonate a legal UE $i$ by computing a forged $D_{ij}$ in the region $j$. The system operator can prevent this kind of attacks by adopting trust platform modular~(TPM) or Trust Execution Environment~(TEE)~\cite{IoTTEE_2017}, which are popular technologies to protect secret keys and compute the related cryptographic operations with the secret keys in dedicated hardware chipset. Additionally, even if the adversary has $GK_{j}$, she/he cannot pass the region-based fast handover authentication as it requires to send out the correct $pID^*_i$, which will be renewed after every communication session. $pID^*_i$ is considered as additional secret shared among the UE $i$ and eNB/HeNB in the region $j$. Hence, the proposed {\bf ReHand} achieves the security of shared group key against key exposure attacks.}

\noindent{{\bf Practical User Anonymity in Mobile Network.} Many prior arts have shown their solutions~\cite{RoamAuth_LLLLZS14,RoamAuth_LCCHAZ15} for user anonymous authentication in mobile networks, where the footprint of communication sessions of UE is unlinkable. In the 3GPP standards of mobile networks, each UE has to complete the attachment to the new visiting serving networks with the assistance of the belonging MME and HSS/AuC so that the mobile service provider of the UE can locate its area in case of an incoming voice call. The above two works adopt group signatures~\cite{GS_BBS04,GS_2012,VLR_GS_BS04,VLR_GS_CLHZ12} to achieve strong user anonymity by which the recipient of the authentication information, made by group signatures with the issued group signing key of distinct user, can only verify the legitimacy by the group public key. That is, only the legal users can produce the valid group signatures on the selected messages by the group signing keys issued from the trust authority. However, strong anonymity causes the impossibility of tracing the location of UEs and results in the failure of incoming voice call service and short message service~(SMS). Thus, AKE protocols with strong user anonymity only check the legitimacy of UEs without knowing the exact identity information for services. For the mobile network services requiring the exact identity information to respond the incoming call, which is exactly for the specified UE, the system should supports traceability for conditional user anonymity.\\
\noindent{\bf Revocation Costs for User Anonymity.}
User revocation to anonymous AKE is essential to the membership management for accountability in mobile networks. Even if the strong user anonymity can achieve identity untraceability against system operators, the revocation of unsubscribed users should get effected to ensure the accountability. The operations of revocation should be considered as parts of operations in AKE. Hence, the costs evaluation of AKE should include revocation costs in practice. 
}
%%==================================================

\section{Security Analysis}

\subsection{{Security Analysis of Initial Handover}}

\begin{theorem}\label{theorem:MemberJoin}
The proposed mutual authentication and key exchange protocol {$\Pi$} with {user anonymity} for the {\bf Initial Handover} phase is secure based on the IND-CCA security of the underlying {pseudorandom function} and the pseudorandom permutation.
\end{theorem}

\begin{lemma}[Mutual Authentication in {\bf Initial Handover}]\label{Lemma:MemberJoinMutualAuth}
The proposed protocol {$\Pi$} in the {\bf Initial Handover} phase is a secure mutual authentication protocol for {UE} $i$ and the system ({including HeNB, eNB, MME, and HSS/AuC}) under the assumption that the underlying pseudorandom permutation $\Omega$ is with PRP security.
\end{lemma}

\begin{proof}
%{\bf Part I:} 
Assume that $\varepsilon$ {is the probability of breaking the mutual authentication security of $\Pi$ by a probabilistic} polynomial time adversary $E$ {with the simulator $\Gamma$, who simulates $\Pi$ and plays}  the game in {\bf Definition II.6} {to break the PRP security}. $\Gamma$ {interacts with $E$ by simulating either UE $i$ or the system with the given $(\Omega,\Omega^{-1})$} by a challenger $\psi$ . {It is either a pair of a pseudorandom permutation and its inverse, i.e., $(E_{K_i},D_{K_i})$, or a truly random permutation and its inverse, i.e., $(\omega,\omega^{-1})$. $\Gamma$ simulates $\Pi$ by the following oracle functions to capture the capability of $E$, where $A$ denotes the UE and $B$ denotes the system in $\Pi$.}
\begin{itemize}
\item[-] $Execute({\Pi^s_{A, B}, \Pi^t_{B, A}})$: {This oracle models the execution of $\Pi$ and outputs the transcripts of the execution.} In order to perform the protocol $\Pi$ correctly, {$\Gamma$ simulates $\Pi^s_{A, B}$ and $\Pi^t_{B, A}$ by querying $\Omega$ and $\Omega^{-1}$ for the encryption of given messages and the decryption of given ciphertexts. Hence, $\Gamma$ can simulate $\Pi$ as described in Fig.\ref{fig:fig_MJ1} successfully.}
\item[-] $Send(\Pi^s_{A, B}, m)$: {This oracle models the capability of active attackers, who sends the message $m$ to $\Pi^s_{A,B}$. If $m = \{C_2\}$, $\Pi^s_{A, B}$ will decrypt it by $\Omega^{-1}$ to check if the format $\{rID^{*}_i,D_{ij},T_{ex}\}$ is correct or not. If so, the oracle outputs {\bf accept} to accept this session, it outputs {\bf reject} to reject this session.}
\item[-] $Send(\Pi^t_{B, A}, m)$: {If $m=\{pID_i,C_1\}$, $\Pi^{t}_{B,A}$ will decrypt $C_1$ by $\Omega^{-1}$ to extract $pID_i$ and $d$. If the format of the extracted $\{rID_,d\}$ is correct, $\Pi^{t}_{B,A}$ computes $C_2=\Omega(TID_{ij},D_{ij},T_{ex},d,pID^{*}_{i})$ by the given $\Omega$ and output $C_2$.}
\end{itemize}

{\bf Challenge:} {Before the challenge phase, $\Gamma$ can query $\Omega$ and $\Omega^{-1}$ with the polynomial number of messages as inputs and receive the corresponding outputs for the training besides the training in the above protocol simulation.} Then, {$\Gamma$ sends a chosen message $\rho=\{pID_i,d\}$ or $\rho=\{TID_{ij},D_{ij},T_{ex},d,pID^{*}_i\}$, which are the messages in the simulation of executing $\Pi$ with $E$, to $\psi$ and} $\psi$ randomly chooses a bit $\theta \in \{0, 1\}$. If $\theta = 0$, then $\psi$ encrypts {$\rho$} by {$\Omega$}; otherwise, $\psi$ encrypts $\rho$ by $\omega$. $\psi$ outputs $\pi$ to $\Gamma$.

{\bf Guess:} {If $E$ can send out $\{pID_i,C_1\}$ or $\{C_2\}$ correctly by acting as a legal UE or system, }, $\Gamma$ will output $\theta' = \theta$; otherwise $\Gamma$ will output $\theta' \in \{0, 1\}$ randomly. 

{When $\theta=0$, the above experiment is a real experiment and $E$ has additional advantage to break the mutual authentication of $\Pi$. When $\theta=1$, the above experiment is a random experiment and $E$ has no advantage to break the mutual authentication of $\Pi$.} {Thus, $\Gamma$ can only exploit the  advantage $\varepsilon$ of breaking the mutual authentication of $\Pi$ by $E$} to break {the security of} PRP, i.e., $E_{K_i}$ and $D_{K_i}$. $\tau$ {is the probability} of breaking $(E_{K_i},D_{K_i})$. {Hence, we have the following.}
\begin{equation}
 \tau {\geq} Pr[\theta' = \theta] -\dfrac{1}{2} = \varepsilon + \dfrac{1}{2}(1 - \varepsilon) - \dfrac{1}{2} = \dfrac{\varepsilon}{2}
\end{equation}
{Since $\tau$ is negligible based on the assumption of PRP security, $\varepsilon$ is also negligible. Hence, the probability of breaking the mutual authentication of $\Pi$ in {\bf member join} phase is negligible.}
\end{proof}

%%==================================================

\begin{lemma}[Key Exchange in {\bf Initial Handover}]\label{Lemma:MemberJoinKeyExchange}
The proposed protocol in the {\bf Initial Handover} phase is a secure key exchange protocol {$\Pi$} if the adopted underlying pseudorandom function $\Lambda$ is with PRF security.
\end{lemma}

\begin{proof}

Assume that $\varepsilon$ {is the probability of breaking the key exchange security of $\Pi$ by}  $E$ {with the simulator $\Gamma$, who simulates $\Pi$ and plays}  the game in {\bf Definition II.6} {to break the PRF security}. $\Gamma$ {interacts with $E$ by simulating either UE $i$ or the system with the given $\Lambda$}. {It is either a pseudorandom function, i.e., $H(d,.)$, or a truly random permutation, i.e., $\lambda$, according to a random bit $\theta\in\{0,1\}$. $\Gamma$ simulates $\Pi$ by the following oracle functions to capture the capability of $E$, where $A$ denotes the UE and $B$ denotes the system in $\Pi$.}

\begin{itemize}
\item[-] $Execute(\Pi^s_{A, B}, \Pi^t_{B, A})$: In order to perform the protocol $\Pi$ correctly, $\Pi^s_{A, B}$ and $\Pi^t_{B, A}$ can query {$\Lambda$ for the keying hash of given messages.}
\item[-] $Send(\Pi^s_{A, B}, m)$: {This oracle can decrypt $C_1$ and output the corresponding $C_2$ successfully since $K_i$ is selected by $\Gamma$.}
\item[-] $Send(\Pi^t_{B, A}, m)$: {This oracle can generate $C_1$ and decrypt $C_2$ successfully since $K_i$ is selected by $\Gamma$.}
\item[-] $Reveal(\Pi^s_{A, B})$: {This oracle outputs the session key held by} $\Pi^s_{A, B}$.
\item[-] $Reveal(\Pi^t_{B, A})$: {This oracle outputs the session key held by} $\Pi^t_{B, A}$.
\item[-] $Test(\Pi^s_{A, B})$: When $E$ makes a $Test$ query to $\Pi^s_{A, B}$, $\Pi^s_{A, B}$ responds a real {session} key $CK$, {which is produced by querying $\Lambda$ with $pID^{*}_i$} or a random string $\alpha$. Note that this query is valid only when the real session key {is not} revealed.
\item[-] $Test(\Pi^t_{B, A})$: When $E$ makes a $Test$ query to $\Pi^t_{B, A}$, $\Pi^t_{B, A}$ responds a real {session} key $CK$, {which is produced by querying $\Lambda$ with $rID_i^{*}$} or a random string as the corresponding answer $K$. Note that this query is {valid only} when the real session key {is not} revealed.
\end{itemize}

{\bf Challenge:} {Before the challenge phase, $Gamma$ can query $\Lambda$ with polynomial number of messages as inputs and receive the corresponding outputs for the training beside the training in the above protocol simulation. Then, $\Gamma$ sends $pID^*_i$ to $\psi$ and}  $\psi$ {randomly chooses a} bit $\theta \in \{0, 1\}$. If $\theta = 0$, $\psi$ computes $pID^*_{i}$ by $H(d,.)$; otherwise, $\psi$ {computes $pID^*_{i}$ by $\lambda$}. 

{\bf Guess:} If neither $\Pi^s_{A, B}$ nor $\Pi^t_{B, A}$ receives the $Test$ query, $\Gamma$ will output $\hat{\theta} \in \{0, 1\}$ at random. Otherwise, if $E$ {queries} $Test$ to $\Pi^s_{A, B}$ or $\Pi^t_{B, A}$, {$\Gamma$} will select a random bit $\theta\in\{0,1\}$ and respond $CK$ to $E$ if $\hat{\theta}=0$. Otherwise, $\Gamma$ responds $\alpha$ to $E$ if $\hat{\theta}=1$. Then, $E$ outputs $\hat{\theta'}=0$ if it guesses the received output of $Test$ query is the real session key. Otherwise, $E$ outputs $\hat{\theta'}=1$. {When $\theta=0$, the above experiment is a real experiment and $E$ has additional advantage $\epsilon$ to break the key exchange security of $\Pi$. When $\theta=1$, the experiment is a random experiment and $E$ has no advantage to break the key exchange security of $\Pi$. Hence, if $\hat{\theta}=\hat{\theta'}$, $\Gamma$ will output $\theta=0$. Otherwise, $\Gamma$ outputs $\theta=0$ or $1$ randomly.}

{$\tau$ is the probability} of breaking {the PRF security of $H(d,.)$. From the above, we have that} 
\begin{equation} 
\tau \geq Pr[\theta' = \theta] -\dfrac{1}{2} = \varepsilon + \dfrac{1}{2}(1 - \varepsilon) - \dfrac{1}{2} = \dfrac{\varepsilon}{2}.
\end{equation}
Thus, the {probability, $\varepsilon$,} of event $Distinguish_{CK}(k)$ of the constructed {session} key $CK$ by an adversary $E$ is negligible {since $\tau$ is also negligible based on the PRF security assumption.}

\end{proof}

\begin{lemma}[User Anonymity in {\bf Initial Handover}]\label{Lemma:MemberJoinUserAnonymity}
	The proposed protocol $\Pi$ in the {\bf Initial Handover} phase is with user anonymity if the adopted underlying pseudorandom permutation $\Omega$ is with PRP security.
\end{lemma}

\begin{proof}
{Assume that $\varepsilon$ is the probability of breaking the user anonymity of $\Pi$ by $E$ with the simulator $\Gamma$, who simulates $\Pi$ and plays  the game in {\bf Definition II.6} to break the PRP security. $\Gamma$ interacts with $E$ by simulating either UE $i$ or the system with the given $\Lambda$. It is either a pseudorandom function, i.e., $H(d,.)$, or a truly random permutation, i.e., $\lambda$, according to a random bit $\theta\in\{0,1\}$. $\Gamma$ simulates $\Pi$ by $Send$, $Execute$, $Reveal$, and $Test$ to capture the capability of $E$, which are the same as that in {\bf Lemma~\ref{Lemma:MemberJoinKeyExchange}}. Besides that, $\Gamma$ additionally simulate the following oracles.
\begin{itemize}
	\item[-] $RevealID(\Pi^s_{A, B})$: This oracle reveals the new anonymous identity $rID^*_i$ held by $\Pi^{s}_{A,B}$.
	\item[-] $RevealID(\Pi^{t}_{B,A})$: This oracle reveals $rID^*_i$ held by $\Pi^{t}_{B,A}$.
	\item[-] $TestID(\Pi^s_{A, B})$: This oracle responds a real $rID^{*}_{i}$ or a random strong depending on the random bit selected by $\Gamma$.
\end{itemize}
	{\bf Challenge:} $\Gamma$ computes $C_2$ with $\{TID_{ij},D_{ij},T_{ex},d,pID^{*}_{i}\}$ by $\Omega$. $\Omega$ given by $\psi$ is $H_{k}$ for $\theta=0$ or a random permutation $\omega$ for $\theta=1$.
	{\bf Guess:} If $E$ output can guess out the given string is a real $pID^*_i$ or a random string after querying $TestID$, then $\Gamma$ will output $\theta=0$. Otherwise, $\Gamma$ outputs $\theta\in\{0,1\}$ randomly. $E$ has additional advantage to break user anonymity when $\Omega$ is a pseudorandom permutation~($E_k$), i.e., $\theta=0$. When $\Omega$ is a random permutation, $E$ has no advantage to break user anonymity. Hence, we have that
	\begin{equation}
	 \tau \geq Pr[\theta' = \theta] -\dfrac{1}{2} = \varepsilon + \dfrac{1}{2}(1 - \varepsilon) - \dfrac{1}{2} = \dfrac{\varepsilon}{2}.
	\end{equation}
	$\tau$ is negligible based on PRP security assumption. Hence, the probability of breaking user anonymity $\varepsilon$ is also negligible.
}
\end{proof}

According to {\bf Lemma~\ref{Lemma:MemberJoinMutualAuth}} {\bf Lemma~\ref{Lemma:MemberJoinKeyExchange}}, and {\bf Lemma~\ref{Lemma:MemberJoinUserAnonymity}} , {\bf Theorem~\ref{theorem:MemberJoin}} holds.\\

\subsection{{Security Analysis of Region-based Fast Handover}}
%%==================================================	
\begin{theorem}\label{theorem:FastHandoverAuth}
{The proposed mutual authentication and key exchange protocol $\Pi'$ with user anonymity for the } {\bf Region-based Fast Handover} phase of the proposed scheme {is secure based on the PRF security}. 
\end{theorem}

\begin{lemma}\label{Lemma:FastHandoverAuth_MutualAuth}
	{The proposed protocol $\Pi'$ in the {\bf Region-based Fast Handover} phase is a secure mutual authentication protocol for UE $i$ and the HeNB under the assumption that the underlying pseudorandom function $\Lambda'_1$ and $\Lambda'_2$ are with PRF security.}
\end{lemma}
\begin{proof}
	{$\Gamma$ simulates $\Pi'$ with the given $\Lambda'=\{\Lambda'_1,\Lambda'_2\}$ by $\psi$, where $\Lambda'_1=H(D'_{ij},.,.,.)$ and $\Lambda'_2=H(K_{He,i},.)$ for $\theta=0$, or $\Lambda'_1$ and $\Lambda'_2$ are the random functions $\lambda'_1$ and $\lambda'_2$ for $\theta=1$. Then, $\Gamma$ can simulate $\Pi$ by the oracle functions, $Execute$ and $Send$, to capture the capability of $E$, which are the same as that in the proof of {\bf Lemma~\ref{Lemma:MemberJoinMutualAuth}}.}
	{In {\bf Challenge} phase, $\Gamma$ will send the message $\{r_{u},r_{h},K_{He,i}\}$ or $r_h$ to $\Lambda'_1$ or $\Lambda'_2$ for the corresponding output, $\delta$ or $\delta'$.}
	
	{In {\bf Guess} phase, $\Gamma$ outputs $\theta'=0$ if $E$ outputs the correct $\delta$ or $\delta'$ successfully. Otherwise, $\Gamma$ outputs $\theta'\in\{0,1\}$ randomly. When $\theta=0$, the above experiment is a real experiment and $E$ has additional probability $\varepsilon'$ to break the mutual authentication security of $\Pi'$. When $\theta=1$, the experiment is a random experiment and $E$ has negligible probability to break mutual authentication security. Hence, we have that
	\begin{eqnarray}
	&\tau'_1 \geq Pr[\theta=\theta']=\varepsilon+\frac{1}{2}(1-\varepsilon)-\frac{1}{2}=\frac{\varepsilon}{2}.\\
	&\tau'_2 \geq Pr[\theta=\theta']=\varepsilon'+\frac{1}{2}(1-\varepsilon')-\frac{1}{2}=\frac{\varepsilon'}{2}.\\
	&\tau'_1+\tau'_2 \geq\frac{\varepsilon+\varepsilon'}{2}.
	\end{eqnarray}
	$\tau'_1$ is the probability of breaking the PRF security of $H(D'_{ij},.,.,.)$, $\tau'_2$ is the probability of breaking the PRF security of $H(K_{He,i},.)$, $\varepsilon$ is the probability of breaking the mutual authentication security by sending the correct $\delta$, and $\varepsilon'$ is the probability of breaking the mutual authentication security by sending the correct $\delta'$. 
	}
\end{proof}

\begin{lemma}\label{Lemma:FastHandoverAuth_KeyExchange}
	{The proposed protocol $\Pi'$ in the {\bf Region-based Fast Handover} phase is a secure key exchange protocol for UE $i$ and the HeNB under the assumption that the  the underlying pseudorandom function $\Lambda'$ is with PRF security.}
\end{lemma}
\begin{proof}
	{
	Assume that $\varepsilon$ is the advantage of breaking the key exchange of $\Pi'$ by $E$.  $\Gamma$ simulates $\Pi'$ with the given $\Lambda$ by $\psi$, where $\Lambda=H(D_{ij},.)$ if $\theta=0$; otherwise, $\Lambda$ is a random function. Then, $\Gamma$ can simulate $\Pi'$ by the oracle functions, $Execute$, $Send$, $Reveal$, and $Test$ to capture the capability of $E$, which are the same as that in the proof of {\bf Lemma~\ref{Lemma:FastHandoverAuth_KeyExchange}}. In {\bf Challenge} phase, $\Gamma$ will send the message $\{r_{u},r_{h}\}$ to $\Lambda$ for the corresponding output $K_{He,i}$.}

	{ 
	In {\bf Guess} phase, $\Gamma$ will output $\theta=0$ or $\theta=1$ randomly if $E$ does not query $Test$. If $E$ query $Test$, $\Gamma$ will respond a real session key $K_{He,i}$ for $\hat{\theta}=0$. Otherwise, $\Gamma$  will respond a random string $\alpha$. $\Gamma$ outputs $\theta=0$ if $E$ outputs $\hat{\theta'}=\hat{\theta}$ correctly. Otherwise, $\Gamma$ outputs $\theta=0$ or $1$ randomly. Hence, we have that
	\begin{equation}
	\tau \geq Pr[\theta' = \theta] -\dfrac{1}{2} = \varepsilon + \dfrac{1}{2}(1 - \varepsilon) - \dfrac{1}{2} = \dfrac{\varepsilon}{2}.
	\end{equation}
	$\tau$ is the probability of breaking the PRF security of $H(D_{ij},.)$. $\epsilon$ is negligible since $\tau$ is also negligible. Therefore, $\Pi'$ guarantees key exchange security based on the PRF security assumption.
	}
\end{proof}

\begin{lemma}\label{Lemma:RegionFastUserAnonymity}
	The proposed protocol $\Pi$ in the {\bf Region-based Fast Handover} phase is with user anonymity if $k>a\times{}b - b$, where $k$ is the number of $\{bR^{I}_j\}_{I\in[1,k]}$, $b$ is the number of UEs within the region $j$, and $a$ is the number of communication sessions launched by a UE in average.
\end{lemma}

\begin{proof}
	The region-based fast handover guarantees the user anonymity by $TID_{ij}=rID_{ij}\oplus{}bR^{I}_{j}$. Since the number of sessions launched by a UE is $a$ and the number of UEs within a region is $b$, there will be $(a\times{}b)$ tuples of $TID_{ij}$ collected by adversaries. If the number of $bR^{I}_{k}$ is $k > a\times{}b-b$, the system should be able to guarantee that $(k+b) > a\times{}b$. That is, the number of variables, including the anonymous identities of UEs and the blind factors, are always more than the equations provided by $TID_{ij}$'s. Hence, there should be no only solution for each anonymous identity for adversaries to link the communication sessions to any specific UE. Hence, the region-based fast handover ensures user anonymity.
\end{proof}

\section{{Comparisons}}
{This section compares the performance in communication and computation costs of ReHand with the three prior arts~\cite{RoamAuth_LLLLZS14,RoamAuth_LCCHAZ15,RoamAuth_HCG15}, which are also with mutual authentication, session key exchange, user anonymity, conditional traceability, and active revocation for the security requirements of roaming-based AKE in 5G. }

\begin{table*}[thb]
	\centering
	\begin{scriptsize}
		\caption{Comparisons on Computation and Communication Costs}
		\begin{tabular}{|L{3cm}|M{3cm}|M{3cm}|M{3cm}|M{3cm}|}
			\hline
			\multicolumn{5}{|c|}{{\bf Computation Costs}}\\\hline\hline
			&  {\bf UE}  & {\bf System} & {\bf User Tracing~(on System)} & {\bf Revocation Costs~(on System)} \\\hline
			CPAL~\cite{RoamAuth_LLLLZS14}&$3T_{e}+10T_{me}$&$T_{e}+7T_{me}+T_{p}$& $4T_{me}+2T_{p}$&$4|RL_{t}|\cdot{}T_{m}+|RL_{t}|\cdot{}(T_{me}+T_{e})+T_{Inv}$\\\hline
			Time-bound Auth~\cite{RoamAuth_LCCHAZ15}& $49T_{e}+8T_{p}$&$46T_{e}+6T_{p}$ & $|RL|\cdot{}T_{e}$& $|RL|\cdot{}T_{e}$\\\hline
			HashHand~\cite{RoamAuth_HCG15}&$T_{pH}+T_{H}+T_{p}$ & $T_{pH}+2T_{H}+T_{p}$ & $T_{pH}+T_{H}+T_{p}$ & N/A\\\hline
			Our scheme~(ReHand)& $\alpha_{R}\cdot{}(2T_{SE}+4T_{H})+(1-\alpha_{R})\cdot{}3T_{H}$&$\alpha_{R}\cdot{}(2T_{SE}+5T_{H})+(1-\alpha_{R})\cdot{}5T_{H}$&0&$|RL^{S_t}_{j}|\cdot{}T_{H}$\\\hline\hline
			\multicolumn{5}{|c|}{\pbox{16.5cm}{\vspace{0.05cm}{\bf Communication Costs~(per authentication)}}} \\\hline\hline
			%&\multicolumn{2}{|c|}{{\bf Initial Handover}}&\multicolumn{2}{|c|}{{\bf Fast Handover}}\\\hline
			CPAL&\multicolumn{4}{|c|}{$\{C_{\alpha}\times{}(15L_{\mbb{G}}+L_{T})\}+\{\frac{1}{T_{RL}}\times{}(C_{\alpha}+C_{\beta})\times{}L_{\mbb{G}}\}+(C_{\beta}\times{}3L_{\mbb{G}})$}  \\\hline
			Time-bound Auth&\multicolumn{4}{|c|}{$C_{\alpha}\times{}\{(11L_{\mbb{G}}+13L_{p}+L_{ID}+L_{H})+2L_{\mbb{G}}\}$} \\\hline
			HashHand&\multicolumn{4}{|c|}{$\{C_{\alpha}\times(2L_{ID}+L_{N}+2L_{H})\}+\{C_{\beta}\times{}(2L_{ID}+L_{N}+L_{H})\}$}\\\hline
			\multirow{2}{3cm}{Our Scheme~(ReHand)}&\multicolumn{4}{|c|}{$\alpha_{R}\times{}\{C_{\alpha}\times{}(3L_{ID}+3L_{K}+L_{T})+C_{\beta}\times{}(3L_{ID}+4L_{K}+L_{T})\}+$}\\
			&\multicolumn{4}{|c|}{$(1-\alpha_{R})\times{}\{C_{\alpha}\times{}(L_{ID}+2L_{H}+2L_{N}+L_{T})+\frac{1}{T_{RL}}\times{}C_{\beta}\times{}|RL^{S_t}_{j}|\times{}L_{H^{Nyb}}\}$}\\\hline\hline
			\multicolumn{5}{|l|}{\pbox{16.5cm}{\vspace*{0.1cm}$T_{SE}$: the computation time of symmetry-based encryption/decryption~(AES)(with an input of 128-bit)\\
					$T_{H}$: the computation time of one-way hash function~(SHA-256)(with an input of 128-bit)\\
					$T_{e}$: the computation time of exponential operation in $\mbb{G}$\\
					$T_{m}$: the computation time of multiplicative operation\\
					$T_{p}$: the computation time of pairing operation\\
					$T_{me}$: the computation time of multi-exponential operation in $\mbb{G}$\\
					$T_{pH}$: the computation time of hash-to-point operation in~($\{0,1\}^{*}\rightarrow{}\mbb{G}$)\\
					$T_{Inv}$: the computation time of inverse operation in $\mbb{G}$\\
					The computation costs on UE side:\\
					$T_{SE}=6.8\times 10^{-3} ms$, $T_{H}$=0.006 ms, $T_{e}$=$T_{Inv}$=70.1 ms, $T_{p}$=135.5 ms, $T_{me}=1.5T_{e}=105.15$ ms, $T_{pH}=10.2$ ms.\\
					The computation costs on system side:\\
					$T_{e}$=$T_{Inv}$= 9.505 ms, $T_{m}$=9.556 ms, $T_{p}$= 5.065ms, $T_{me}=1.5T_{e}=$ 14.257 ms, $T_{pH}=1.413$ ms.\vspace*{0.1cm}}} \\ 
			\multicolumn{2}{|l}{\pbox{10cm}{ $L_{ID}$: the length of an identity~(128-bit)\\
					$L_{H}$: the length of a hash~(SHA-256) output\\
					$L_{\mbb{G}}$: the length of an element in $\mbb{G}$\\
					$L_{p}$: the length of an element in $\mbb{Z}_p$~($p$ is a prime)
			}} &
			\multicolumn{3}{l|}{\pbox{10cm}{$L_{N}$: the length of a nonce\\
					$L_{K}$: the length of a symmetry-based secret key\\
					$L_{T}$: the length of timestamp\\
					$L_{H^{Nyb}}$: the length of a Nyberg one-way accumulate hash containing one item}} \\
			\multicolumn{5}{|l|}{\pbox{16.5cm}{\vspace*{0.1cm}$L_{ID}=L_N=L_H=L_K=128$ bits, $L_{\mbb{G}}=170$ bits, $L_{p}=171$ bits, $L_T=64$ bits, $L_{H^{Nyb}}$~(per 10 accumulated items) = 722.33 bits, $L_{H^{Nyb}}$~(per 100 accumulated items) = 1444.66 bits
			}}\\
			\hline
		\end{tabular}
		\label{tab: performance}
	\end{scriptsize}
\end{table*}
\subsection{{Computation and Communication Costs}}
We evaluate the computation/communication costs of the proposed ReHand protocols with the other roaming-based authentication protocols~\cite{RoamAuth_LLLLZS14,RoamAuth_LCCHAZ15,RoamAuth_HCG15}  on a smartphone of ASUS Zenfone 3 as a UE testbed. The smartphone runs Android 8.0.0 mobile operating system and is equipped with 2.0 GHz octa-core ARM Cortex-A53 CPU and 3GB RAM.  The cryptographic libraries for the implementation of the required cryptographic operations in the proposed scheme and the related works are java pairing based cryptography (JPBC)~\cite{ISCC_DecIov11} and Java Cryptography Extension~(JCE)~\cite{JCA}. The evaluation also run the above cryptographic libraries on a APPLE Macbook Pro (2016 model) with 2.9 GHz dual-core Intel Core i5 CPU and 8GB RAM for the estimation of computation costs. Table~\ref{tab: performance} shows the total computation cost and the communication costs of the proposed schemes with the three related works~\cite{RoamAuth_LLLLZS14,RoamAuth_LCCHAZ15,RoamAuth_HCG15}, and the definitions of the computation times of all operations and the message lengths of all variables.

Regarding the message lengths, the length of an element from $\mbb{G}_1$ is 170 bits and from $\mbb{G}_T$ is 340 bits for the pairing mapping by MNT curves~\cite{PB_MNT01} for 80 bits security. The key lengths, $L_{K}$, of the symmetry-based encryption~(i.e., Advanced Encryption Standard, AES), keyed hash function~(i.e., Secure Hash Algorithm 2, SHA-2), and their outputs are 128 bits. The lengths of an identity and nonce used in AKE protocol are also 128 bits for the consistency. 
	
In order to evaluate the communication costs among the proposed scheme and the other prior arts empirically, we assume that the eNB/HeNB as visiting authentication node~(VAN), which is the closest authentication node to UE, and the HSS/AuC as home authentication node~(HAN), which is the authentication node of the belonging home network of UE. Without loss of generality for communication costs, $C_{\alpha}$ denotes the communication time of a unit~(i.e., 512-bit as a minimum data frame) between UE and VAN. $C_{\beta}$ denotes the communication time of a unit between VAN and HAN. Since the communication cost between eNB and HeNB is extremely low as the wired X2 interface is used between them. Besides, the performance metrics define the rate of roaming to a HeNB in a new visiting region as $\alpha_{R}$ and the rate of roaming to a HeNB in a visited region as $(1-\alpha_{R})$ to estimate the effect of handover on the performance by time of the proposed protocols, i.e., {\bf Initial Handover} and {\bf Region-based Handover}. In the testbed of the communication evaluation, an WiFi access point~(AP) of D-Link DIR-612 N300 with a laptop of APPLE Macbook Pro~(2016 model) emulates eNB/HeNB, and two guest operating systems~(OSs) of Debian/Linux 9 on the virtual machine of Google Cloud Platform Computer Engine emulate VAN and HAN, respectively. By running each communication experiment for 10,000 times, $C_{\alpha}$ is 4.36 millisecond~(ms) and $C_{\beta}$ is 261.76~ms.

\subsection{{Revocation Costs}}
{Revocation check is essential to verify the legality of the membership in the system and considered as a part of authentication. Hence, the costs of revocation check should be evaluated in the performance comparison. Since the revocation list can be updated for every fixed period to reduce the size, the size of revocation list is defined as $|RL_t|$. For the proposed scheme, the size of the revocation list in a region $j$ at the specific period $S_t$ is defined as $|RL^{S_t}_{j}|$. In order to compare the performance unbiased, we also assume that the revocation lists are updated periodically in the other related works~\cite{RoamAuth_LLLLZS14,RoamAuth_LCCHAZ15,RoamAuth_HCG15}. In~\cite{RoamAuth_LCCHAZ15}, the revocation list cannot be updated for only specific period of times. Hence, we denote the size of the complete revocation list as $|RL|$. Nevertheless, the communication costs are also affected by the update frequency of revocation list. Thus, we define the period of updating revocation list as $T_{RL}$ and the frequency as $\frac{1}{T_{RL}}$. The performance evaluation will take the above defined variables related to revocation check into account.}

\begin{figure*}[!ht]
	\centering
	\includegraphics[scale=0.48]{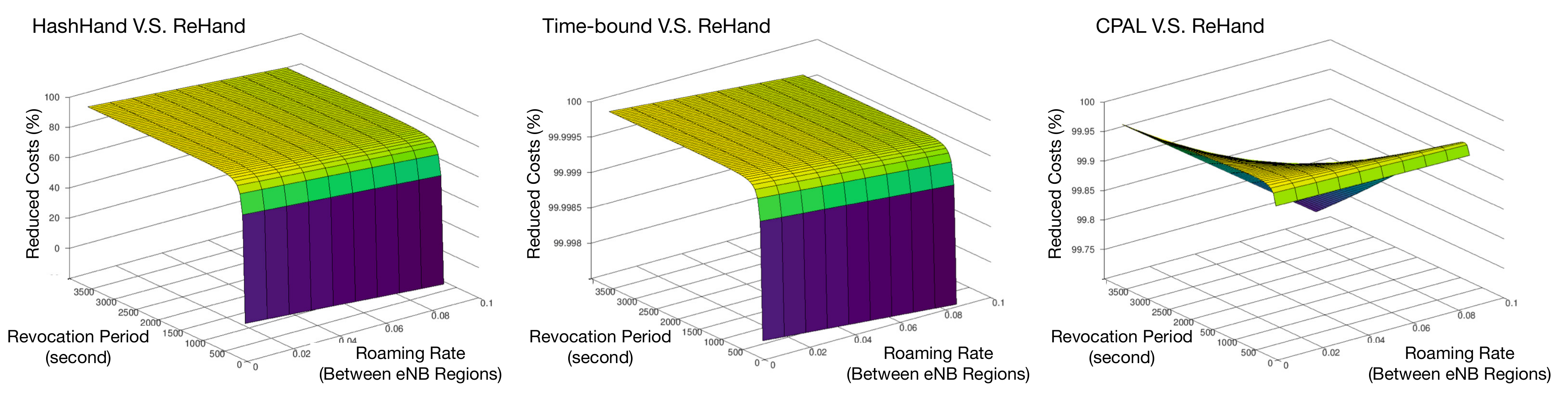}
	\caption{Performance Comparisons}
	\label{fig:performance_comparison}
\end{figure*}
\subsection{{Performance Evaluation}}
In order to evaluate the effect of roaming to a new region on the performance, the time latency of the proposed protocol is $\mathrm{T}_{\mathbf{H-AKE}} = \alpha_{R} \times \mathrm{T}_{\mathbf{I-AKE}}+(1-\alpha_{R})\times\mathrm{T}_{\mathbf{F-AKE}}$, where $\mathrm{T}_{\mathbf{I-AKE}}$ is the time of {\bf Initial Handover} and $\mathrm{T}_{\mathbf{F-AKE}}$ is the time of {\bf Region-based Fast Handover}. 
 
In this evaluation, we assume that the number of revoked UEs is 1,000,000~\footnote{According to the statistics reports from National Communication Commission~(NCC) Taiwan, the number of base stations~(eNB), $N_{eNB}$, is 22,000 and the number of revoked subscribers~(i.e., $|RL|$) is around 1,000,000 in the major telecommunication company, Chunghwa Telecom.}, the period of updating revocation lists~(i.e., $T_{RL}$) is from 60 seconds to 3,600 seconds~(1 hour), the range of speed, $v$, of UE is from 0 to 500 kilometer per hour~(KM/h), and the diameter, $r$, of eNB is 2 KMs. Here, the expiration time, $T_{exp}$, of each warrant is the same as $T_{RL}$ in ReHand. $\alpha_{R}$ is affected by the expiration of the warrant and roaming to a new eNB region and defined as 
\begin{equation}
\alpha_{R}=
\begin{cases}
\frac{1}{T_{exp}} + \frac{v}{r\times{}3,600}, & \text{if } \alpha_{R} < 1 \\
1, & \text{otherwise.}
\end{cases}
\end{equation}
Based on the given $v\in[0,500]$~KM/h, $r=2$~KMs, and  $T_{exp}\in[60,3600]$, the range of $\alpha_{R}$ is from $2.78\times{}10^{-4}$ to $8.6\times{}10^{-2}$. 
Figure~\ref{fig:performance_comparison} shows the performance comparison~(computation and communication costs) of the proposed ReHand scheme with the three prior arts~\cite{RoamAuth_LLLLZS14,RoamAuth_LCCHAZ15,RoamAuth_HCG15}, the ReHand greatly reduces the cost to $82.92$\% compared with HashHand scheme~\cite{RoamAuth_HCG15} when $T_{RL}\geq{}240$~(i.e., $\geq{}2$ minutes), to $99.99$\% compared with Time-bound scheme~\cite{RoamAuth_LCCHAZ15}, and to $99.95$\% compared with CPAL scheme~\cite{RoamAuth_LLLLZS14}.

\section{Conclusion}
%As the next generation of mobile communication networks, 5G is expected to conform to the requirements of super-high data rate, ultra-low latency, large capacity, power efficiency, and security. We pointed out the differences of infrastructure and communication between 4G and 5G and found that the handover protocols in 4G are not suitable for 5G. The total latency will increase rapidly along with the high frequency of handover operations processed by the mobile users if 5G adopts the handover protocols of 4G. 
{This work} proposes a {region-based secure fast handover framework} that is not only tailored to the technical direction of small cell network in 5G, but also combines the properties of secure mutual authentication, privacy preservation, computation efficiency, and functional active revocation. {The proposed authentication framework adopts the techniques of group key, one-time identity, and accumulated one-way hash, so that every authentication within the same region only involve UE and the visiting HeNB.} We also provide formal security analysis {to demonstrate the proposed security scheme is secure based on the cryptographic hard problems.} Compared to {the} other works, {the proposed} scheme {eliminates considerable} computation {and communication} costs {for small cell networks in 5G}.

\appendices

\ignore{
% use section* for acknowledgment
\section*{Acknowledgment}
This work was partially supported by the Ministry of Science and Technology of the Taiwan under grant MOST 104-2221-E-110-043, MOST 105-2923-E-110-001-MY3, and Aim for the Top University Plan of the National Sun Yat-sen University and Ministry of Education, Taiwan, R.O.C.
}

% Can use something like this to put references on a page
% by themselves when using endfloat and the captionsoff option.
%\ifCLASSOPTIONcaptionsoff
%  \newpage
%\fi

% trigger a \newpage just before the given reference
% number - used to balance the columns on the last page
% adjust value as needed - may need to be readjusted if
% the document is modified later
%\IEEEtriggeratref{8}
% The "triggered" command can be changed if desired:
%\IEEEtriggercmd{\enlargethispage{-5in}}

% references section

% can use a bibliography generated by BibTeX as a .bbl file
% BibTeX documentation can be easily obtained at:
% http://www.ctan.org/tex-archive/biblio/bibtex/contrib/doc/
% The IEEEtran BibTeX style support page is at:
% http://www.michaelshell.org/tex/ieeetran/bibtex/
%\bibliographystyle{IEEEtran}
% argument is your BibTeX string definitions and bibliography database(s)
%\bibliography{IEEEabrv,../bib/paper}
%
% <OR> manually copy in the resultant .bbl file
% set second argument of \begin to the number of references
% (used to reserve space for the reference number labels box)

\bibliographystyle{IEEEtran}  
\small\bibliography{nsysuthesis}
\ignore{

}

% biography section
% 
% If you have an EPS/PDF photo (graphicx package needed) extra braces are
% needed around the contents of the optional argument to biography to prevent
% the LaTeX parser from getting confused when it sees the complicated
% \includegraphics command within an optional argument. (You could create
% your own custom macro containing the \includegraphics command to make things
% simpler here.)
%\begin{IEEEbiography}[{\includegraphics[width=1in,height=1.25in,clip,keepaspectratio]{mshell}}]{Michael Shell}
% or if you just want to reserve a space for a photo:
\ignore{
\begin{IEEEbiography}[{\includegraphics[width=1in,height=1.25in,clip,keepaspectratio]{figs/CIFan.eps}}]{Chun-I Fan}
received the MS degree in computer science and information engineering from National Chiao Tung University, Taiwan, in 1993, and the PhD degree in electric engineering at National Taiwan University in 1998. From 1999 to 2003, he was an associate researcher and project leader of Telecommunication Laboratories, Chunghwa Telecom Co., Ltd, Taiwan. In 2003, he joined the faculty of the department of computer science and engineering, National Sun Yat-sen University, Kaohsiung, Taiwan, and has been a full professor since 2010. His current research interests include applied cryptology, cryptographic protocols, information and communication security, and he has published over 100 technical papers. He won the Dragon PhD Thesis Award from Acer Foundation, Best PhD Thesis Award from Institute of Information \& Computing Machinery in 1999, Best Student Paper Awards
in National Conference on Information Security 1998. He advised his
graduate students to win the Best Student Paper Awards in National
Conference on Information Security 2007, Best Master Thesis Award from
Taiwan Association for Web Intelligence Consortium in 2011, Outstanding
Master Dissertation Award from Taiwan Institute of Electrical and Electronic
Engineering in 2011 and 2012, Master Thesis Award from Chinese
Cryptology and Information Security Association in 2012, and Outstanding
PhD Dissertation Award from Institute of Information \& Computing
Machinery in 2012. He also was the editor-in-chief of Information Security
Newsletter and is an executive director of Chinese Cryptology and
Information Security Association.
\end{IEEEbiography}

% if you will not have a photo at all:
\begin{IEEEbiography}[{\includegraphics[width=1in,height=1.25in,clip,keepaspectratio]{figs/jjhuang.eps}}]{Jheng-Jia Huang}was born in Kaohsiung, Taiwan.
He received the M.S. degree in information management from National
Kaohsiung First University of Science and Technology, Kaohsiung,
Taiwan, in 2012. He now is a Ph.D. student in computer science and
engineering from National Sun Yat-sen University, Kaohsiung, Taiwan,
His current research interests include cloud computing and security,
social network security and authentication, network and
communication security, information security, and applied
cryptography.
\end{IEEEbiography}

% insert where needed to balance the two columns on the last page with
% biographies
%\newpage

\begin{IEEEbiography}[{\includegraphics[width=1in,height=1.25in,clip,keepaspectratio]{figs/MZZhong.eps}}]{Min-Zhe Zhong}
received the MS degree in computer science and information engineering from National Sun Yat-sen University, Kaohsiung, Taiwan, in 2014. His research interests include information security, network and communication security, cryptographic protocols, and applied cryptography.
\end{IEEEbiography}

\begin{IEEEbiography}[{\includegraphics[width=1in,height=1.25in,clip,keepaspectratio]{figs/Richard.eps}}]
	{\bfseries Ruei-Hau Hsu} (M'15) received the B.S. and M.S. degrees in Computer Science from Tunghai University, Taiwan, in 2002 and 2004, respectively. He received the Ph.D. degree in Computer Science and Engineering at National Sun Yat-sen University, Kaohsiung, Taiwan, in 2012. He was the postdoctoral research fellow at the Department of Computer Science, National Chiao Tung University from 2012 to 2014, and at iTrust, Centre for Research in Cyber Security at Singapore University of Technology and Design established in collaboration with Massachusetts Institute of Technology~(MIT) from 2014 to 2017. Currently, he is a Scientist with Data Storage Institute~(DSI), Agency for Science, Technology and Research~(A*STAR), Singapore.
	
	Dr.~Hsu received two Best Doctoral Dissertation Awards from Institute of Information and Computing Machinery and Best Doctoral Dissertation Award from Chinese Cryptology and Information Security in 2012, respectively. From August to December 2007, he joined the International Collaboration for Advancing Security Technology~(iCAST) program as a visiting scholar at Carnegie Mellon University~(CMU), America. From June to September 2010 and March 2011 to February 2012, he earned scholarships, granted by Deutscher Akademischer Austausch Dienst~(DAAD) Germany and National Science Council~(NSC) Taiwan, as a visiting scholar at Center for Advanced Security Research Darmstadt~(CASED) in Technische Universitat Darmstadt. In 2012, he has been the member of the Phi Tau Phi Scholastic Honor Society. He also served as a Publication Co-Chair of The 16th IEEE sponsored Asia-Pacific Network Operations and Management Symposium~(APNOMS) 2014.
\end{IEEEbiography}

\begin{IEEEbiography}
[{\includegraphics[width=1in,height=1.25in,clip,keepaspectratio]{figs/wtchen.eps}}]
{Wen-Tsuen Chen}
received the B.S. degree in nuclear engineering
from National Tsing Hua University, Taiwan, and the M.S. and Ph.D.
degrees in electrical engineering and computer sciences both from
University of California, Berkeley, in 1970, 1973, and 1976,
respectively. He has been with the National Tsing Hua University
since 1976 and is a Distinguished Chair Professor of the Department
of Computer Science. He has served as Chairman of the Department,
Dean of College of Electrical Engineering and Computer Science, and
the President of National Tsing Hua University. Since March 2012, he
has joined the Academia Sinica, Taiwan as a Distinguished Research
Fellow of the Institute of Information Science. His research
interests include computer networks, wireless sensor networks,
mobile computing, and parallel computing.

Dr. Chen received numerous awards for his academic accomplishments
in computer networking and parallel processing, including
Outstanding Research Award of the National Science Council, Academic
Award in Engineering from the Ministry of Education, Technical
Achievement Award and Taylor L. Booth Education Award of the IEEE
Computer Society, and is currently a lifelong National Chair of the
Ministry of Education, Taiwan. Dr. Chen is the Founding General
Chair of the IEEE International Conference on Parallel and
Distributed Systems and the General Chair of the 2000 IEEE
International Conference on Distributed Computing Systems among
others. He is an IEEE Fellow and Fellow of the Chinese Technology
Management Association.
\end{IEEEbiography}

}

% You can push biographies down or up by placing
% a \vfill before or after them. The appropriate
% use of \vfill depends on what kind of text is
% on the last page and whether or not the columns
% are being equalized.

%\vfill

% Can be used to pull up biographies so that the bottom of the last one
% is flush with the other column.
%\enlargethispage{-5in}

% that's all folks
\end{document}